\documentclass[conference]{IEEEtran}
\IEEEoverridecommandlockouts
\usepackage{cite}
\usepackage{amsmath,amssymb,amsfonts}
\usepackage{algorithmic}
\usepackage{graphicx}
\usepackage{textcomp}
\usepackage{booktabs}
\usepackage{xcolor}
\usepackage{amsthm}
\usepackage{breqn}
\usepackage{bm}
\usepackage[ruled,linesnumbered]{algorithm2e}
\usepackage{subcaption}
\newtheorem{definition}{Definition}

\newtheorem{lemma}{Lemma}
\usepackage[font=small]{caption} 
\newtheorem{theorem}{Theorem}
\usepackage[backref]{hyperref} 

\def\BibTeX{{\rm B\kern-.05em{\sc i\kern-.025em b}\kern-.08em
    T\kern-.1667em\lower.7ex\hbox{E}\kern-.125emX}}

\begin{document}
\title{
    When Metaverses Meet Vehicle Road Cooperation: Multi-Agent DRL-Based Stackelberg Game for Vehicular Twins Migration}
    
\author{
    \IEEEauthorblockN{Jiawen Kang, \textit{Senior Member, IEEE},
        Junhong Zhang,
        Helin Yang,
        Dongdong Ye, and \\
         M. Shamim Hossain, \textit{Senior Member, IEEE}
    }
    \thanks{
        J. Kang, J. Zhang, and D. Ye are with the Guangdong University of Technology, China (e-mail: kavinkang@gdut.edu.cn; junhong1013@163.com; dongdongye8@163.com).

        H. Yang are with the School of Informatics, Xiamen University, China (e-mail: helinyang066@xmu.edu.cn).

        M. Shamim Hossain is with the Department of Software Engineering, College of Computer and Information Sciences, King Saud University, Riyadh 12372 (email: mshossain@ksu.edu.sa).

        The work was presented in part at the 2023 IEEE International Conference on Distributed Computing Systems Workshops (ICDCSW) (\textit{*Corresponding author: M. Shamim Hossain})
    }
}

\maketitle

\begin{abstract}
    Vehicular Metaverses represent emerging paradigms arising from the convergence of vehicle road cooperation, Metaverse, and augmented intelligence of things. Users engaging with Vehicular Metaverses (VMUs) gain entry by consistently updating their Vehicular Twins (VTs), which are deployed on RoadSide Units (RSUs) in proximity. The constrained RSU coverage and the consistently moving vehicles necessitate the continuous migration of VTs between RSUs through vehicle road cooperation, ensuring uninterrupted immersion services for VMUs. Nevertheless, the VT migration process faces challenges in obtaining adequate bandwidth resources from RSUs for timely migration, posing a resource trading problem among RSUs. In this paper, we tackle this challenge by formulating a game-theoretic incentive mechanism with multi-leader multi-follower, incorporating insights from social-awareness and queueing theory to optimize VT migration. To validate the existence and uniqueness of the Stackelberg Equilibrium, we apply the backward induction method. Theoretical solutions for this equilibrium are then obtained through the Alternating Direction Method of Multipliers (ADMM) algorithm. Moreover, owing to incomplete information caused by the requirements for privacy protection, we proposed a multi-agent deep reinforcement learning algorithm named MALPPO. MALPPO facilitates learning the Stackelberg Equilibrium without requiring private information from others, relying solely on past experiences. Comprehensive experimental results demonstrate that our MALPPO-based incentive mechanism outperforms baseline approaches significantly, showcasing rapid convergence and achieving the highest reward.
\end{abstract}


\section{Introduction}

\IEEEPARstart{T}{he} Metaverse, deemed the "successor to the mobile Internet," seamlessly merges physical and virtual spaces, driven by advanced technologies including Extended Reality (XR), Digital Twin (DT), Augmented Intelligence of Things (AIoT), and haptic technologies \cite{xu2022full}. Vehicular Metaverses implement the Metaverses within autonomous vehicles and intelligent roads and are expected to revolutionize the vehicle road cooperation systems \cite{zhou2022vetaverse}. For instance, Panasonic unveiled its AR HUD (head-up displays) 2.0 system in 2022 that uses artificial intelligence to display lane edges, objects on the road, and other information through windshields and side windows, thereby enhancing the driving experience with greater intuitiveness and satisfaction. In Vehicular Metaverses, each Vehicular Metaverse User (VMU) connects with its corresponding Vehicular Twin (VT) to access the virtual space for immersive services like AR navigation, virtual meetings, and Massive Multiplayer Online
(MMO) games. The VT acts as a vital bridge between the physical and virtual spaces throughout the vehicle and VMU life cycle, providing more intelligent and efficient services \cite{xu2022full,hu2022review}. Equipped with onboard sensors such as radars, cameras, and location receivers, VMUs continually collect real-time information, including current location, historical trajectory, and service preferences, generating VT tasks for synchronizing the VT \cite{zhang2023learningbased}. Based on the information provided by VMUs, VTs can ask for Metaverse service preparation and recommendations from RoadSide Units (RSUs) and the cloud.

However, on the one hand, the locally limited computation resources of VMUs may pose challenges in executing computation-intensive VT tasks and updating results to RSUs \cite{ren2022quantum}. Consequently, VMUs prefer offloading these VT tasks to proximal RSUs through vehicle road cooperation. On the other hand, owing to the constrained RSU coverage and the consistently moving vehicles, VMUs progressively move away from their VTs. As a result, it is necessary to migrate VTs between RSUs to ensure uninterrupted Metaverse services as vehicles move, called “Vehicular Twins migration” \cite{zhang2023learningbased}. To ensure VT migration efficiency, the migration source RSU (Metaverse Service Provider, MSP) needs to purchase sufficient bandwidth resources from the migration destination RSU (Metaverse Resource Provider, MRP) to facilitate VT migration. In the resource trading process, MSPS must comprehend the pricing trends of MRPs, allowing them to make informed decisions regarding bandwidth demand aligned with their specific requirements. Likewise, MRPs may strategically adjust prices in response to resource demand, aiming to foster sustainable growth within their customer base. Hence, the design of an appropriate incentive mechanism is crucial to facilitate trading between MSPs and MRPs.

Motivated by the sequential decision-making between MSPs and MRPs, we formulate a game-theoretic incentive mechanism to capture the distributed decision-making features of these players. Specifically, we model the interaction between MSPs and MRPs as a Stackelberg game with multi-leader multi-follower, with MRPs as leaders and MSPs as followers\cite{nie2020multi,huang2022joint,xu2021privacy}. In addition to exploring trading issues between MSPs and MRPs, we delve into positive social effects among MSPs and incorporate them in the game modeling of players. Metaverse services often share similar functions and content \cite{chu2023dynamic}. For instance, popular AR mobile games like Pokémon Go and Pikmin Bloom currently utilize similar functionalities offered by the Google Maps API. By strategically migrating VTs that request the same Metaverse service to a common destination, the resource consumption at the destination can be significantly reduced, enhancing the service experience for VMUs. Furthermore, traditional methods like the backward induction method \cite{huang2022joint} for solving the Stackelberg game require all players to provide all their private information, which is unreasonable in reality. Rational players in the game tend to safeguard their privacy while seeking maximum profit. Fortunately, recent advancements in Deep Reinforcement Learning (DRL) offer a promising approach to learning optimal strategies based solely on experience, eliminating the need for prior information. Consequently, DRL can effectively derive the Stackelberg Equilibrium (SE) without disclosing any player's private information \cite{xu2021multiagent, xu2021privacy, zhang2023learningbased}.

In this paper, we propose a multi-leader multi-follower game-theoretic incentive mechanism for VT migration, integrating social awareness and queueing theory. Initially, we first outline a comprehensive VT migration process and calculate the migration delay based on the queueing theory and edge computing. By integrating social effects into the players' utility, we formulate the problem as a Stackelberg game. The existence and uniqueness of the SE are then proven using the backward induction method, and the theoretical solution is derived through the Alternating Direction Method of Multipliers (ADMM) algorithm. Finally, to protect the privacy of players, we transform the proposed game into a multi-agent partially observable Markov decision process and design a \underline{M}ulti-\underline{A}gent \underline{L}STM-based \underline{P}roximal \underline{P}olicy \underline{O}ptimization (MALPPO) to determine optimal solutions. Simulation results demonstrate the superior performance of our MALPPO algorithm over baselines, approaching the SE.
As far as we know, we are the first to jointly consider social effects and privacy protection in the exploration of incentive mechanisms for the Metaverse. The key contributions can be summarized as follows:
\begin{itemize}
    \item Given the constrained RSU coverage and the consistently moving vehicles, we proposed a VT migration-assisted Vehicular Metaverses framework ensuring the continuous immersion Metaverse service based on the queueing theory.
    \item To improve the VT migration efficiency between MSPs and MRPs, we formulate a game-theoretic incentive mechanism with multi-leader multi-follower considering the social effects. We then prove the existence and uniqueness of the SE through backward induction and employ the ADMM algorithm to derive the theoretical solution.
    \item  We design a multi-agent deep reinforcement learning algorithm named MALPPO to address incomplete information and security concerns in solving the Stackelberg game. Our numerical results highlight the effectiveness of our learning-based approach, showing quick convergence to the SE and superior performance compared to baseline schemes such as vanilla MAPPO and MAA2C.
\end{itemize}
The structure of the paper unfolds as follows: Section II furnishes an overview of related works. Section III introduces the system model, while Section IV engages in the discussion of SE. Moving on to Section V, we delve into the proposed DRL-based solution. Section VI presents the simulation results and analysis. Section VII summarizes this paper.

\section{Related Works}
\subsection{Metaverse}
The origin of the Metaverse can be traced back to Neal Stephenson's 1992 science fiction novel, Snow Crash. Since then, propelled by advancements in cutting-edge technologies, the Metaverse has garnered significant attention from both industry and academia. Notably, in 2021, Facebook's rebranding to "Meta" underscored a transformative shift from a mere "social media company" to a dedicated "Metaverse company," emphasizing its commitment to Metaverse development. Xu et al.'s study \cite{xu2022full} explores challenges and solutions in the edge-enabled Metaverse, focusing on communication, networking, computation, and blockchain for immersive access, aligning with the Metaverse vision. Zhou et al. \cite{zhou2022vetaverse} introduce the term “Vetaverse” signaling an envisioned convergence between vehicular industries and the Metaverse. Numerous studies concentrate on specific technologies within the Metaverse.
For instance, Hu et al. \cite{hu2022review} systematically detail the construction of the driver's digital twin, emphasizing key aspects and technologies for understanding both external states (e.g., distraction, drowsiness) and internal states (e.g., intentions, emotions, trust). Du et al. \cite{du2023ai} propose an efficient information-sharing scheme employing full-duplex device-to-device semantic communications to address the limited computation power in Vehicular Mixed Reality Metaverses.

\subsection{Incentive Mechanism in Metaverse}
Recent research in the Metaverse has explored various incentive mechanisms, each tailored to different optimization objectives. In \cite{jiang2022reliable}, a hierarchical game-based coded distributed computing framework is proposed, focusing on optimizing real-time rendering. This framework incorporates reputation-based miner coalition formation and a Stackelberg game incentive mechanism. In \cite{xu2022wireless}, researchers address challenges related to user matching and resource pricing between VR users and service providers. Their solution incorporates a Double Dutch auction mechanism along with deep reinforcement learning. In \cite{doe2023promoting}, a contract-theoretic incentive mechanism is introduced to motivate users, addressing resource challenges in operating full nodes and promoting blockchain network sustainability and growth. Additionally, \cite{ren2022quantum} proposes a quantum collective learning and many-to-many matching scheme in the Metaverse for autonomous vehicles. While significant progress has been made in studying incentive mechanisms for the Metaverse, the majority of existing works overlook the VT migration issue arising from the consistently moving vehicles and neglect the privacy protection problem in optimization. Consequently, designing an incentive mechanism suitable for VT migration remains a challenging task.

\subsection{Reinforcement Learning Based Stackelberg Game}
Reinforcement Learning (RL) has proven to be a powerful tool in achieving equilibrium solutions in the Stackelberg game without requiring participants to disclose their privacy. In \cite{yao2019resource}, the authors apply the "WoLF-PHC" multi-agent RL approach to model cloud provider-miner interactions. Zhan et al. \cite{zhan2020learning} use a policy optimization approach in a Stackelberg game for federated learning on IoT. In our prior work \cite{zhang2023learningbased}, we propose a learning-based Stackelberg game considering only a single leader for VT migration in vehicular metaverses. However, these studies primarily focus on a Stackelberg game in a single-provider scenario, potentially leading to a monopoly market and reduced competitiveness among resource providers, diverging from the reality of multiple providers. In \cite{zhao2023multi}, a multi-leader-follower Stackelberg game incentivizes edge devices in federated edge learning using an MA-DADDPG algorithm. In \cite{xu2021privacy}, Xu et al. the interactions between providers and IoT devices are formulated as a multi-leader-follower Stackelberg game, and an RL-based algorithm is designed to determine the optimal pricing strategy with privacy preservation. However, most of the existing works ignore the cooperativeness and social effects between the followers. Building upon all the above limitations of existing works, as far as we know, we are the first to jointly consider the cooperativeness among the followers and competitiveness among the leaders while studying the incentive mechanisms for Metaverse under privacy preservation.

\section{System Model}

\subsection{Migration-empowered Vehicular Metaverses Model}
\begin{figure*}[t]
    \centering{\includegraphics[width=0.85\textwidth]{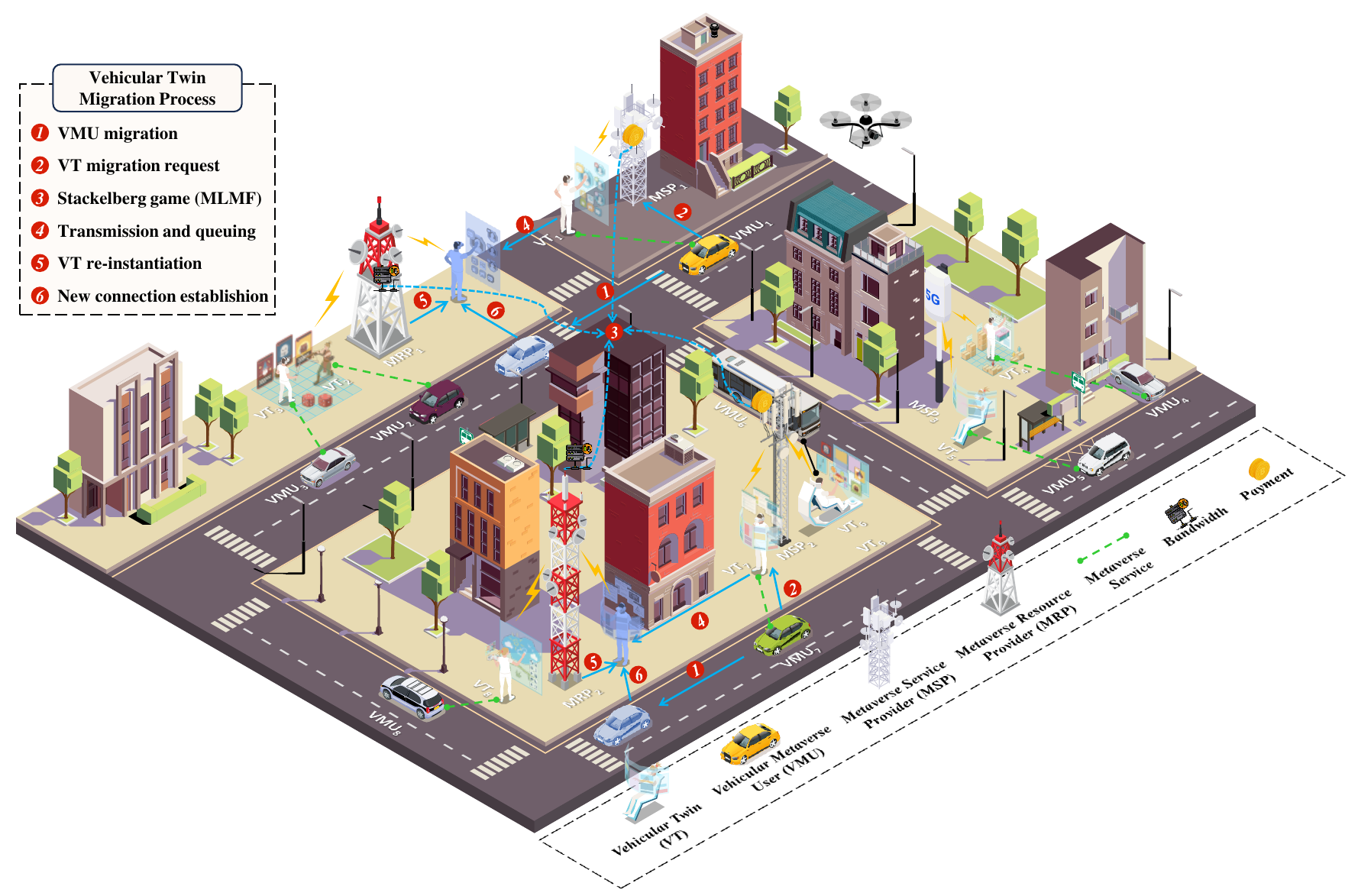}}
    \caption{An illustration of the multi-leader multi-follower Stackelberg game-based vehicular twin migration framework.}\label{framework}
\end{figure*}

The system model of Migration-empowered Vehicular Metaverses (MVM) is illustrated in Fig. \ref{framework}. In this framework, each VMU accesses the Metaverse through its corresponding VT to obtain immersive services. The VT is a digital twin covering the lifecycle of the VMU. The VMU gathers real-time information, including location and service preferences, through its sensors (such as cameras and radars) to generate VT tasks and synchronize the VT. Due to the limited local computation and storage capabilities of VMUs, it is challenging for them to execute computation-intensive VT tasks. Consequently, VMUs choose to offload VTs onto nearby RSUs, leveraging their abundant resources to perform large-scale VT synchronization tasks.

However, due to the constrained RSU coverage and the consistently moving vehicles, VMUs progressively move away from their VTs. As a result, it is necessary to migrate VTs between RSUs to maintain continuous Metaverse services as vehicles move. In the MVM, the MSPs are responsible for collecting sensing data from VMUs to construct and maintain their VTs. To ensure the immersive service of the VMU, each MSP purchases bandwidth resources from the target RSUs (i.e., MRPs) and migrates the VMU's VT when the VMU is about to leave the service range of it. When a VT arrives at the MRP, it is added to the MRP processing queue waiting for re-instantiation. It is noteworthy that the MRP transforms into a new MSP after completing the VT migration.

\subsection{Migration Task and Delay Model}
In this study, we consider the engagement of $M$ MRPs and $N$ MSPs in VT migration and resource trading, represented by the sets $\mathcal{N}=\{1,\ldots,n,\ldots, N\}$ and $\mathcal{M}=\{1,\ldots,m,\ldots, M\}$. The vehicular twin migration task of MSP $i$ is represented as $J_{i}=\left\{D_{i}, L_i, K _{i}^{\max }\right\}$, where $D_i$ is the total migrated VT data including the vehicle configuration, historical interaction data, and real-time VT states, the $L_i$ represents the number of CPU cycles required by re-instantiating VT, and $K _{i}^{\max }$ denotes the maximum tolerant delay \cite{zhang2023learningbased}. To guarantee the quality of the immersive Metaverse service, the total migration delay must be less than $K_i^{\max}$. To overcome the interference among MSPs during the twin migration task, we adopt the Orthogonal Frequency Division Multiplexing Access (OFDMA) technology in the MVM \cite{huang2022joint} to guarantee that all transmission channels between MSPs and MRPs remain orthogonal. Given the bandwidth $b_{ij}$ purchased by MSP $i$ from the MRP $j$, the transmission data rate can be calculated as $r_{ij}=b_{ij} \log _{2}\left(1+\frac{\rho^{} h d^{-\varepsilon}}{N_{0}}\right)$, where $\rho$, $h$, $d$, $\varepsilon$, and ${N}_{0}$ represent the transmitter power of MSP, the unit channel power gain, the distance between the MSP and the MRP, the path-loss coefficient, and the additive white Gaussian noise (AWGN) power, respectively\cite{li2019stackelberg,zhang2023learningbased}.

According to the vehicular twin migration process in Fig. \ref{framework}, the migration delay contains three parts: the transmission delay, the queue delay, and the re-instantiation delay. 1) The transmission delay from MSP $i$ to MRP $j$ can be formulated as $t_{ij}^{tran}=\frac{D_i}{r_{ij}}$. 2) Based on the queuing theory \cite{li2019stackelberg}, we consider that each MRP re-instantiation queue is assumed as an M/M/1 queue model, the queue delay is calculated as $t_{ij}^{que}=\frac{\lambda_{j}}{\mu_{j}(\mu_{j}-\lambda_{j})}$, where $\mu_{j}$ represents the processing rate of MRP $j$ and $\lambda_j$ denotes the computation task arriving rate. 3) The re-instantiation delay is formulated as $t_{ij}^{com}=\frac{L_i}{f_j}$, where $f_j$ is the computation capability of edge server on MRP $j$. Therefore, the total migration delay of task $J_i$ can be formulated as
$T_{ij} = t_{ij}^{tran} + t_{ij}^{que} + t_{ij}^{com}$.

\subsection{Utility Modeling of MSPs and MRPs}
Motivated by the previous work \cite{xu2021privacy,huang2022joint}, we consider that each MSP follows a probabilistic decision-making model to select the MRP for bandwidth. The pairing probability $\theta_{ij}$ between MSP $i$ and MRP $j$ is proportional to the reciprocal of the bandwidth selling price of MRP $j$, which can be expressed as $\theta_{ij}=\frac{\frac1{p_j}}{\sum_{l\in\mathcal{M}}\frac1{p_l}}$,
where $p_j > 0$ represents the bandwidth selling price of MRP $j$. To clarify, we first introduce some symbols here. The selling price of all MRPs is defined as a vector $\boldsymbol{p}=\left\{p_{j}\right\}_{j \in \mathcal{M}}$ and the bandwidth demand of all MSPs is defined as a vector $\boldsymbol{b}=\left\{b_{i}\right\}_{i \in \mathcal{N}}$, where $b_{i} = \left\{b_{ij}\right\}_{j \in \mathcal{M}}$ represents the vector of bandwidth purchasing from all MRPs.
\subsubsection{The Utility Function of MSP}
MSPs procure bandwidth resources from MRPs with the aim of minimizing migration transmission delay and enhancing the seamless immersive experience during VT migration. Consistent with approaches such as \cite{candogan2012optimal, nie2020multi}, we define the immersion function for MSP $i$ as a linear-quadratic function: $S(b_{i}) = \alpha_ib_{ij}-\beta_ib_{ij}^2$. Here, $\alpha_i$ represents the maximum internal satisfaction degree, and $\beta_i$ denotes the sensitivity when purchasing $b_{ij}$ bandwidth, where $\alpha_i > 0$ and $\beta_i > 0$. It's worth noting that this concave linear-quadratic function implies decreasing marginal returns, signifying that satisfaction benefits increase at a diminishing rate. After finishing VT migration, the MSP pays the MRP according to the service level agreement (SLA). We define the payment function as $C(b_{i}) = b_{ij}p_j$.

Additionally, we focus on the positive social effect among MSPs \cite{candogan2012optimal,li2019stackelberg,chu2023dynamic}, that is, one MSP's increased bandwidth purchases can have a beneficial impact on his peers. For instance, VTs engaged in AR mobile games like Pokémon Go and Pikmin Bloom can share the same Metaverse service function as Google Maps API. This sharing reduces resource consumption at the destination, ultimately enhancing the service experience for all MSPs. The social ties among MSPs are represented as an adjacency matrix $\mathbf{W}=[w_{ik}]~{i,k\in\mathcal{N}}$, where the entry in the $i$ th row and $k$ th column, $w_{ik}$, signifies the influence that MSP $i$ has on MSP $k$. This matrix is calculated based on historical average interactions and remains fixed throughout this paper. Similar to \cite{li2019stackelberg}, we assume that the influences between MSP $i$ and MSP $k$ are mutual and identical, i.e., $w_{ik} = w_{ki}$. A higher value of $w_{ik}$ signifies a closer relationship and a more significant influence on each other's decisions \cite{nie2020multi}. Consequently, given the bandwidth demand profile for the remaining providers except MSP $i$, we model this positive social effect as the external benefits for MSP $i$, denoted as $\Phi (b_{i},b_{-i}) = \sum_{k\in\mathcal{N}}w_{ik}b_{ij}b_{kj}$.
To sum up, the utility of MSP $i$ is formulated as follows:
\begin{equation}
    \label{MSP}
    U_{F_i} = \sum_{j\in\mathcal{M}} \frac{\frac1{p_j}}{\sum\limits_{l\in\mathcal{M}}\frac1{p_l}} (\alpha_ib_{ij}-\beta_ib_{ij}^2+
    \sum_{k\in\mathcal{N}}w_{ik}b_{ij}b_{kj}-b_{ij}p_j)
\end{equation}
\subsubsection{The Utility Function of MRP}
Given the pairing probability $\theta_{ij}$, the utility of the MRP is determined as the difference between the sum of bandwidth fees paid by paired MSPs and the cost of processing the VT migration task. Hence, we formulate the utility function of MRP $j$ as:
\begin{equation} \label{MRP}
    U_{L_j} =\sum_{i\in\mathcal{N}} \frac{\frac1{p_j}}{\sum_{l\in\mathcal{M}}\frac1{p_l}} (b_{ij}p_j-b_{ij}c_j)
\end{equation}
where $c_j > 0$ is the unit migration cost of bandwidth including transmission and re-instantiation cost. Each MRP $j \in \mathcal{M}$ adjusts its selling price $p_j$ to attract more MSPs purchasing their bandwidth resource and maximizes its own utility $U_{L_j}$. Note that $U_{L_j}$ is influenced by both MSPs and other MRPs, i.e., there is a competitive game among the MRPs.

\subsection{Stackelberg Game Formulation}
In the proposed MVM model, both MSPs and MRPs need to decide on optimal bandwidth purchase strategies and selling price to maximize their respective utilities. Additionally, the optimization problems of MSPs and MRPs are intricately linked. Game theory emerges as a promising tool for studying decentralized decision-making among participants in resource trading scenarios. We formulated the interaction between MSPs and MRPs as a multi-leader and multi-follower (MLMF) Stackelberg game with two stages. In the leader subgame, MRPs decide their selling price of bandwidth resources first. In the follower subgame, MSPs determine their bandwidth demand based on their VT migration tasks and MRPs' unit selling price of bandwidth.
\begin{figure*}[!h]
    \hrulefill
    \begin{align}
         & \delta\mathcal{H}_{ij}(\mathbf{B})>\mathcal{H}_{ij}(\delta\mathbf{B}) \label{scalability}
        =\frac{\delta\alpha_i + \delta\sum_{k\in\mathcal{N}}w_{ik}b_{kj}-\delta p_j}{2\beta _i}
        -\frac{\alpha_i + \delta\sum_{k\in\mathcal{N}}w_{ik}b_{kj}-p_j}{2\beta _i}
        = \frac{(\delta-1) (\alpha_i - p_j)}{2\beta _i} >0.                                          \\
         & \frac{\partial U_{L_j}}{\partial p_{j}}=\label{Uj1}
        \sum_{i=N}\frac{{\alpha_i + \sum_{k\in\mathcal{N}}w_{ik}b_{kj}}{}-2p_{j}-p_j^{2}\sum_{l\in\mathcal{M}/j}\frac{1}{p_l}+c_j+c_j{(\alpha_i + \sum_{k\in\mathcal{N}}w_{ik}b_{kj})}{}\sum_{l\in\mathcal{M}/j}\frac{1}{p_l}}
        {2\beta _i(1+p_j\sum_{l\in\mathcal{M}/j}\frac{1}{p_l})^{2}}                                  \\
         & \frac{\partial^{2} U_{L_j}}{\partial p_{j}^{2}}=\label{Uj2}
        \sum_{i=N}
        \frac{-2-2{(\alpha_i + \sum_{k\in\mathcal{N}}w_{ik}b_{kj})}{}
        \sum_{l\in\mathcal{M}/j}\frac{1}{p_l} - 2c_j\sum_{l\in\mathcal{M}/j}\frac{1}{p_l} -2c_j{(\alpha_i + \sum_{k\in\mathcal{N}}w_{ik}b_{kj})}{}\sum_{l\in\mathcal{M}/j}\frac{1}{p_l}^2}
        {2\beta _i(1+ p_j\sum_{l\in\mathcal{M}/j}\frac{1}{p_l})^{3}}<0
    \end{align}
\end{figure*}
Given the bandwidth price profile for all MRPs $\boldsymbol{p}$ and the bandwidth demand profile of other MSPs (i.e., $\boldsymbol{B}_{-i}$), the follower-level problem is formulated as follows:

\begin{equation}
    \begin{aligned}
        \textbf{Problem 1:}~\max~ & U_{F_i}(b_i,\boldsymbol{B}_{-i},\boldsymbol{P})      \\
        \text{s.t. }              & b_{ij} \ge 0                                         \\
                                  & \sum_{j\in\mathcal{M}}\theta_{ij}T_{ij}\le K^{max}_i \\
    \end{aligned}
\end{equation}
where the total migration delay mentioned should not exceed the maximum tolerance $K^{max}_i$. At the leader-level, the MRP maximizes its utility according to the pricing strategy of all MRPs and the bandwidth purchase strategy of all MSPs as follows:
\begin{equation}
    \begin{aligned}
        \textbf{Problem 2:}~\max~ & U_{L_j}(p_j,\boldsymbol{P}_{-j},\boldsymbol{B}) \\
        \text{s.t. }              & p_j\in\left[c_j,p^{max}\right]                  \\
    \end{aligned}
\end{equation}
where $p^{max}$ is the upper bound of selling price $p_j$. When the selling price set by the MRP exceeds the expectations of MSPs, none of the MSPs are willing to pay for the bandwidth.

Thus, $\textbf{Problem 1}$ and $\textbf{Problem 2}$ jointly constitute a Stackelberg game. Our goal is to find the SE to achieve the optimal solution for the formulated game. In this SE, the utilities of MRPs are maximized, assuming that MSPs obtain their best response. Furthermore, neither any MRP nor any MSP can enhance their individual utility by deviating from their respective strategies \cite{jiang2022reliable,huang2022joint}. The SE for our proposed MVM model is represented as follows:

\begin{definition}
    (Stackelberg Equilibrium, SE): Let $\boldsymbol{B^{*}}=\left\{b_{i}^*\right\}_{i \in \mathcal{N}}$ and $\boldsymbol{P^{*}}=\left\{p_{j}^{*}\right\}_{j \in \mathcal{M}}$ represent the optimal bandwidth demand strategy and the optimal bandwidth selling price, respectively. Let $\boldsymbol{B_{-i}^{*}}$ denote the optimal bandwidth demand strategies of all other MSP except $i$ and $\boldsymbol{P_{-i}^{*}}$ denotes the optimal bandwidth selling price strategies of all other MRP except $j$. Then, the stable point $(\boldsymbol{B^{*}}, \boldsymbol{P^{*}})$ can be SE if the following inequalities is strictly satisfied \cite{huang2022joint,zhang2023learningbased}:
    \begin{equation}
        \left\{\begin{array}{l}U_{F_i}(b_i^{*},\boldsymbol{B_{-i}}^{*},\boldsymbol{P}^{*}) \geq U_{F_i}(b_i,\boldsymbol{B_{-i}}^{*},\boldsymbol{P}^{*}),\:\forall i \in \mathcal{N},\vspace{0.05in} \\
            U_{L_j}(p_j^{*},\boldsymbol{P_{-j}}^{*},\boldsymbol{B}^{*})\geq U_{L_j}(p_j,\boldsymbol{P_{-j}}^{*},\boldsymbol{B}^{*}),\:\forall j \in \mathcal{M}.\end{array}\right.
    \end{equation}
\end{definition}

\section{Equilibrium Analysis for MLMF Stackelberg Game Based on ADMM Algorithm}


\subsection{Analysis of the Follower-level Game} \label{followergame}
In the follower-level game, each MSP $i$ adjusts its bandwidth demand to maximize its utility based on all MRPs' price profiles $\boldsymbol{p}$. The first-order and second-order derivatives of $U_{F_i}$ concerning $b_{ij}$ are derived as follows:

\begin{align}
    \frac{\partial U_{F_i}}{\partial b_{ij}}         & =\frac{\frac1{p_j}}{\sum_{l\in\mathcal{M}}\frac1{p_l}}(\alpha_i-2\beta_ib_{ij}+
    \sum_{k\in\mathcal{N}}w_{ik}b_{kj}-p_j)                                                                                                              \\
    \frac{\partial^{2} U_{F_i}}{\partial b_{ij}^{2}} & =-\frac{\frac1{p_j}}{\sum_{l\in\mathcal{M}}\frac1{p_l}}(2\beta _i)<0\label{follower second-order}
\end{align}
The negative second-order derivative presented in (\ref{follower second-order}) implies the quasi-concavity of the MSP's utility function $U_{F_i}$ with respect to $b_i$. Our proof of the Nash Equilibrium existence among MSPs stands. To establish its uniqueness, we examine the best response function of MSP. By applying the first-order optimality condition $\frac{\partial U_{F_i}}{\partial b_{ij}} = 0$, the best response function for MSP $i$ is expressed as follows:
\begin{equation}
    b_{ij}^{*}=\frac{\alpha_i + \sum_{k\in\mathcal{N}}w_{ik}b_{kj}-p_j}{2\beta _i}
\end{equation}
The follower-level sub-game exhibits a unique Nash Equilibrium if the best response function of MSP adheres to a standard function \cite{xu2021privacy}.
\begin{lemma} \label{standard}
    A function $\mathcal{H}_{ij}(\boldsymbol{B})$ is a standard function if and only if it satisfies the following three conditions:
    \begin{itemize}
        \item Positiveness: $\mathcal{H}_{ij}(\boldsymbol{B}) > 0$
        \item Monotonicity: $\forall \boldsymbol{B^{\prime}} > \boldsymbol{B}$, $\mathcal{H}_{ij}(\boldsymbol{B}^{\prime})>\mathcal{H}_{ij}(\boldsymbol{B})$
        \item Scalability: $\forall\delta >1,\delta\mathcal{H}_{ij}(\boldsymbol{B})>\mathcal{H}_{ij}(\boldsymbol{\delta B})$
    \end{itemize}
\end{lemma}

\begin{theorem}\label{ap}
    If $a_i > p_j$ is satisfied, the sub-game perfect equilibrium in the MSPs' sub-game is unique.
\end{theorem}
\begin{proof}

    Because $b_i \ge 0$, we further obtain the best response function
    \begin{equation} \label{bestMSP}
        b_{ij}^*=\mathcal{H}_{ij}(\boldsymbol{B})=
        \begin{cases}0,         & \psi_{ij}<0    \\
             \psi_{ij}, & \psi_{ij}\ge 0\end{cases}
    \end{equation}
    where we have $\psi_{ij} =  \frac{\alpha_i + \sum_{k\in\mathcal{N}}w_{ik}b_{kj}-p_j}{2\beta _i} $.

    According to \textbf{Lemma} \ref{standard}, if the best response function given in (\ref{bestMSP}) satisfies positiveness, monotonicity, and scalability, we can prove the uniqueness of MSPs' sub-game. Obviously, these three properties are satisfied at the lower bound, i.e., $\mathcal{H}_{ij}(\boldsymbol{B}) = 0$ $(\psi_{ij}<0)$. Then we assess whether $\mathcal{H}_{ij}(\boldsymbol{b}) = \psi_{ij}$ is a standard function. Firstly, due to the Individual Rationality (IR) \cite{wen2023freshness}, each MSP will obtain a non-negative utility while requesting VT migration task, i.e., $U_{F_i} \ge 0$.
    For positiveness, we can easily get $\alpha_i + \sum_{k\in\mathcal{N}}w_{ik}b_{kj}-p_j \ge 0$ with simple transform according to $U_{F_i} \ge 0$. Thus, we have $\mathcal{H}_{ij}(\mathbf{B})>0$ due to the positiveness of $\alpha_i + \sum_{k\in\mathcal{N}}w_{ik}b_{kj}-p_j$ and $\beta_i$. Secondly, we show the monotonicity of $\psi_i$. Let $\boldsymbol{B^{\prime}} > \boldsymbol{B}$, we have $\sum_{k\in\mathcal{N}}w_{ik}b_{kj}^{\prime } > \sum_{k\in\mathcal{N}}w_{ik}b_{kj}$, which proves the monotonicity condition. We consider that the satisfaction coefficient $\alpha_i$ will be no less than the charge of bandwidth resource $p_j$. For scalability, it can be seen from (\ref{scalability}) that $\psi_{ij}$ satisfies scalability according to Theorem \ref{ap} and $\delta >1$. In summary, the best response function of MSP satisfies three properties characteristic of a standard function. Consequently, we prove both the existence and uniqueness of the follower-level Nash Equilibrium.
\end{proof}
\subsection{Analysis of the Leader-level Game}

\begin{theorem}
    The unique Stackelberg Equilibrium denoted as $(\boldsymbol{B^{}}, \boldsymbol{P^{}})$ is established in the formulated game if $a_i > p_j$, where both the bandwidth demand strategies and bandwidth price strategies of the MSPs and MRPs are optimized.
\end{theorem}
\begin{proof}
    After each follower selects the optimal bandwidth demand strategy under any bandwidth price, the MRPs can maximize their utilities by adjusting the optimal $p_j$. By substituting (\ref{bestMSP}) into MRP's utility function (\ref{MRP}), we can derive
    \begin{equation}
        U_{L_j} =\sum_{i\in\mathcal{N}} \frac{\frac1{p_j}}{\sum_{l\in\mathcal{M}}\frac1{p_l}}
        (\frac{\alpha_i + \sum_{k\in\mathcal{N}}w_{ik}b_{kj}-p_j}{2\beta _i} )(p_j-c_j)
    \end{equation}
    By computing the first-order and second-order derivatives of $U_{L_j}$ with respect to $p_j$, we obtain Eq. (\ref{Uj1}) and Eq. (\ref{Uj2}). Due to the positiveness of every term in the numerator of $\frac{\partial^{2} U_{L_j}}{\partial p_{j}^{2}}$, we can easily derive that $\frac{\partial^{2} U_{L_j}}{\partial p_{j}^{2}} < 0$ and $U_{L_j}$ is a concave function. Subsequently, applying the first-order optimality condition $\frac{\partial U_{L_j}}{\partial p_{j}} = 0$ and considering the upper and lower bounds for setting the price, we deduce the best response of MRP $j$ as follows:
    \begin{equation}\label{bestMRP}
        p_{j}^*=\mathcal{F}_j(\boldsymbol{P})=
        \begin{cases}0,            & \varsigma _j<c_j               \\
             \varsigma _j, & 0\leq \varsigma _j\leq p^{max} \\
             p^{max},      & \varsigma_j>p^{max}\end{cases}
    \end{equation}
    where
    \begin{equation}
        \begin{aligned}
            \begin{split}
                \varsigma_j = & \sqrt{\frac{\sum\limits_{i\in\mathcal{N}}  \sum\limits_{\substack{l\in\mathcal{M}\\l\ne j}} \frac{c_j+p_l}{p_l}
                        \left({\alpha_i + \sum\limits_{k\in\mathcal{N}}w_{ik}b_{kj}\sum\limits_{\substack{l\in\mathcal{M}\\l\ne j}} \frac{1}{p_l}
                            +1}\right)-\frac{1}{2\beta_i}^{2}}
                    {\sum\limits_{\substack{l\in\mathcal{M}\\l\ne j}} \frac{1}{p_l}
                        ^{2}\cdot \frac{1}{2\beta_i}^{2}}}
            \end{split}
        \end{aligned}
    \end{equation}

    The non-cooperative leader sub-game attains a unique Nash Equilibrium if $\mathcal{F}_j(\boldsymbol{P})$ follows a standard function. As in Section \ref{followergame}, we assess the three properties: positiveness, monotonicity, and scalability when $c_j\leq \varsigma _i\leq p^{max}$. Due to space limitations in the paper, we omit the description of the proof for this section. For detailed information, please refer to \cite{xu2021privacy}. Thus, we can demonstrate that the MRP's best response function also adheres to a standard function, guaranteeing the uniqueness of the Nash Equilibrium. In summary, we conclusively affirm the existence and uniqueness of the formulated MLMF Stackelberg game, complemented by \textbf{Theorem} \ref{ap}.
\end{proof}

\subsection{ADMM-based Algorithm for SE}
The primary challenge in resolving the formulated multi-leader multi-follower Stackelberg game arises from the high dimensionality of each player's strategy space \cite{xiong2019cloud}. As analyzed in (\ref{bestMSP}) and (\ref{bestMRP}), the optimization of the follower subgame and the leader subgame is coupled. In this section, the ADMM algorithm is introduced to decompose the complex optimization problem into several subproblems for the SE. The following explanation outlines the iterative process of the ADMM-based algorithm for optimizing the utilities of MSPs and MRPs, with algorithm details shown in \textbf{Algorithm 1}.

\textbf{Inner Loop}: In the inner loop, each MSP $i \in \mathcal{N}$ calculates its bandwidth purchasing strategy $b_{ij}$ from each MRP $j$ to maximize its utility, leveraging the announced prices $p_j^{(q)}$ set at the beginning of each iteration $q$, where $q$ is the number of outer loop iterations. In each iteration, the process of MSP's strategy updating can be expressed as:
\begin{equation}
    \begin{aligned}
        b_{i}^{(q)}\left(t+1\right)=
        \arg \max
        U_{F_i}\left(b_{i}^{(q)}\right)+\sum_{j=1}^{\mathcal{M}}\eta_{i}^{(q)}\left(\hat{t} \right) \theta _{ij}T_{ij} \\
        +\frac{\rho}{2}\left\|\sum_{k=1, k \neq \mathrm{j}}^{\mathcal{M}} \theta _{ik}T_{ik}^{(q)}\left(\hat{t} \right)+\theta _{ij}T_{ij}-K_i^{\max }\right\|_{2}^{2}
    \end{aligned}
\end{equation}
where $\hat{t} = t + 1$ when $k < i$ and $\hat{t} = t$ when $k > i$. $||_{2}^{2}$ denotes the Frobenius norm, $\rho$ is the damping factor, and $t$ represents the index of the inner loop step. The Lagrangian multiplier $\eta$ undergoes updates according to the following expression \cite{nie2020multi,xiong2019cloud}:
\begin{equation} \eta_{i}^{(q)}\left(t+1\right)=\eta_{i}^{(q)}\left(t\right)+\rho\left(\sum_{j=1}^{\mathcal{M}} \theta _{ij}T_{ij}^{(q)}(t+1)-K_i^{\max }\right)
\end{equation}
The MSP's strategy $b_{i}$ undergoes updates using the two equations above until both $b_i$ and $\eta_{i}^{(q)}$ exhibit negligible changes.

\textbf{Outer Loop}: In the outer loop, the MRPs adjust their bandwidth pricing strategies according to the feedback of MSPs to maximize their profits. The iterative steps for the MRP $j$ using ADMM are given by $p_{j}^{(q)}\left(t+1\right)=\arg \max (U_{L_j}(p_j^{(q)}))$. Once all MRPs have adjusted their pricing strategies, they notify the MSPs to update their purchasing strategies, and the $(q+1)$ iteration begins. The stop condition of the iteration process is shown as follows:
\begin{equation}
    \left\|\sum_{j\in\mathcal{M}}U_{L_j}(p_j^{(q)})-\sum_{j\in\mathcal{M}}U_{L_j}(p_j^{(q-1)})\right\|\leq\Xi
\end{equation}
where $\Xi$ is a pre-determined threshold.
\begin{algorithm}[t]
    \small
    \caption{ADMM Algorithm for Stackelberg Game with Multi-Leader Multi-Follower}
    \label{ADMMM}
    \SetKwData{Left}{left}
    \SetKwData{This}{this}
    \SetKwData{Up}{up}
    \SetKwFunction{Union}{Union}
    \SetKwFunction{FindCompress}{FindCompress}
    \SetKwInOut{Input}{Input}
    \SetKwInOut{Output}{Output}
    \Input{$p_{j} \in\left[c_j, p_{\max }\right]$, precision threshold $\Xi$, $q=1$; }
    \Output{Optimal purchasing strategy of MSPs $\left(b_{i}\right)^{*}$;\\Optimal pricing strategy of MRPs $\left(p_{j}\right)^{*}$;}
    \While{$\left\|\sum_{j\in\mathcal{M}}U_{L_j}\left(p_{j}^{(q)}\right)-\sum_{j\in\mathcal{M}}U_{L_j}\left(p_{j}^{(q-1)}\right)\right\| \leq \Xi$}{
    Inner loop: Based on the pricing strategies given by MRPs $\boldsymbol{p}$, MSPs iteratively update the purchasing strategies to obtain the amount of bandwidth $\left(b_{i}\right)^{q}$ that maximize $U_{F_i}$;

    Outer loop: Based on the purchasing strategy given by MSPs $\boldsymbol{B}$, the MRPs achieve the optimal bandwidth price $\left(p_{j}\right)^{q}$ that maximize $U_{L_j}$;

    $q = q + 1$;
    }

\end{algorithm}

\section{Multi Agent Deep Reinforcement Learning-Based Approach for Stackelberg Game}
Conventional heuristic algorithms like backward induction, analytic methods, and iterative techniques can be employed to seek the optimal solution in the formulated game \cite{xu2021privacy}. Nonetheless, applying these centralized algorithms in a real game environment is impractical for the following reasons. On the one hand, heuristic algorithms require full information about the game environment. However, due to the non-cooperative relationship, each game player is not willing to disclose its private information such as the migration data size $D_i$, the maximum tolerant delay $K_i^{max}$, the migration cost of bandwidth $c_j$ and so on, for secure concerns. On the other hand, heuristic algorithms are often one-off algorithms with high time complexity. A slight change of the game model will cause the algorithm to re-iterate for the optimal solution\cite{xu2021multiagent}.

Fortunately, Deep Reinforcement Learning (DRL) emerges as a valuable tool for learning optimal policies based on past historical experiences, leveraging the current state and provided rewards without the need for prior information. Moreover, it is challenging to solve the proposed optimization problems using traditional single-agent approaches due to the dynamic environment in Vehicular Metaverses. To obtain the SE practically, we first describe how to transform the proposed MLMF Stackelberg game into a learning task and design a MADRL approach to explore the optimal solution for both MSPs and MRPs.
\begin{figure*}[t]
    \centering{\includegraphics[width=0.90\textwidth]{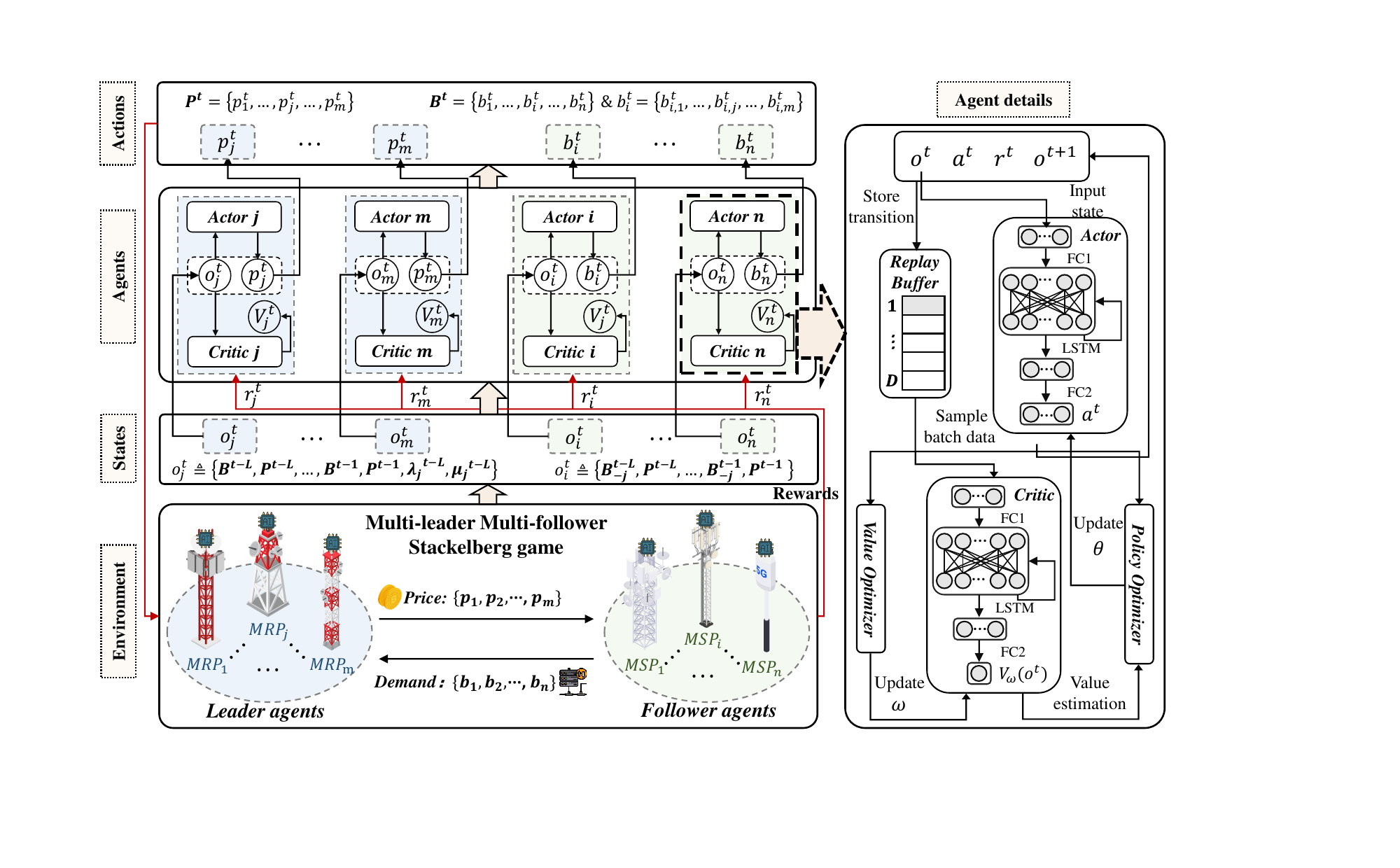}}
    \caption{The workflow of using MALPPO algorithm for MLMF Stackelberg game in vehicular twin migration.}\label{MALPPO}
\end{figure*}

\subsection{Deep Reinforcement Learning for Stackelberg Game}
To leverage DRL, we begin by representing the MLMF Stackelberg game as a multi-agent Partially Observable Markov Decision Process (POMDP) denoted by the tuple $\langle\mathcal{S},\mathcal{O},\mathcal{A},\mathcal{T},\mathcal{R}\rangle $, encompassing the state space, observation space, action space, state transition probability function set, and reward function, respectively. The MSPs and MRPs are intelligent agents competing with each other to maximize their rewards. We describe the detailed definition of each term as follows.

\textbf{\emph{1) State Space:}} In the current decision round $t$ $(t= \small\{1, 2,\ldots, T\small\})$, the state space is defined as a union of the current MRPs' pricing strategies $\boldsymbol{P}^{t}$, the MSPs' bandwidth demand strategies $\boldsymbol{B}^{t}$, the task arriving rate $\Lambda^t = \{\lambda_{1}^{t},\ldots,\lambda_{j}^{t},\ldots,\lambda_{m}^{t}\}$ and the task processing rate $\boldsymbol{\mu}^t = \{\mu_{1}^{t},\ldots,\mu_{j}^{t},\ldots,\mu_{m}^{t}\}$. Formally, $
    S^t\triangleq\left\{\boldsymbol{P}^{t},\boldsymbol{B}^{t},\Lambda^t,\boldsymbol{\mu}^t\right\}
$.

\textbf{\emph{2) Partially Observable Space:} } Due to privacy protection reasons, each agent is unable to observe the complete state of the environment. In VT migration, agents make decisions solely based on their local observations within a formulated partially observable space. At the beginning of each training round $t$, MRP $j$ first decides the pricing policy $p_j^t$ according to the historical pricing strategies of MRPs, the historical bandwidth demand strategies of MSPs in the past $L$ rounds, the task arriving rate $\lambda_j^t$ and the task processing rate $\mu_{j}^t$. Thus, the observation space of MRP $j$ can be represented as $o_j^t\triangleq\left\{\boldsymbol{B}^{t-L},\boldsymbol{P}^{t-L},\ldots,\boldsymbol{B}^{t-1},\boldsymbol{P}^{t-1},\lambda_j^{t},\mu_{j}^{t}\right\}$.

In the second stage, the follower MSP $i$ determines its bandwidth purchase policy $b_i^t$ based on the historical pricing strategies of MRPs and the historical bandwidth demand strategies of the other MSPs, where $\boldsymbol{B}_{-j}^{t} = \{b_{1,1}^{t},\ldots,b_{i-1,m}^{t},\ldots,b_{i+1,1}^{t},\ldots,b_{n,m}^{t}\}$. Thus, the observation space of MSP $i$ is represented as $o_i^t\triangleq\left\{\boldsymbol{B}_{-j}^{t-L},\boldsymbol{P}^{t-L},\ldots,\boldsymbol{B}_{-j}^{t-1},\boldsymbol{P}^{t-1}\right\}$. $\boldsymbol{P}^{t-L}$ and $\boldsymbol{B}^{t-L}$ are randomly generated in the initial phase when $t<L$. This historical experience facilitates the agents in understanding how their strategies evolve and influence the game outcome in the current time slot \cite{zhang2023learningbased}.

\textbf{\emph{3) Action Space:} } Upon receiving the observation $o_j^t$, the MRP agent $j$ is required to undertake a pricing action $p_j^t$ aimed at maximizing utility. Considering the migration cost $c_j$ and the upper bound price $p^{max}$ for the pricing action, the action space is defined as $p_j^t \in [c_j, p^{max}]$. Different from the MRP, the MSP determines a bandwidth demand vector $b_i^{t}$ instead of a single action after receiving the observation $o_i^t$.

\textbf{\emph{4) Reward Function:}} Following the actions taken by all agents, the system state undergoes a transition based on the state transition probability function $\mathcal{T}$. Subsequently, each agent receives an immediate reward corresponding to the current state and the actions taken. The reward functions for MRPs and MSPs, aligned with the utility functions in (\ref{MRP}) and (\ref{MSP}), can be defined as follows:

\begin{align}
    \label{r_MSP}
    R(o_j^t,p_j^{t}) & = U_{L_j}(p_j^t,\boldsymbol{P}_{-j}^t,\boldsymbol{B}^t), \\
    R(o_i^t,b_i^{t}) & =
    \begin{cases}
        U_{F_i}(b_i^t,\boldsymbol{B}_{-i}^t,\boldsymbol{P}^t), & \sum_{j\in\mathcal{M}}\theta_{ij}T_{ij}^t \le K^{max}_i, \\[2ex]
        0,                                                     & \sum_{j\in\mathcal{M}}\theta_{ij}T_{ij}^t >  K^{max}_i.
    \end{cases}
\end{align}
\begin{algorithm}[t]
    \small
    \label{MALPPOcode}
    \caption{MALPPO-based Solution for MLMF Stackelberg Game under privacy protection}
    Initialize maximum episodes $E$, maximum time slots $T$ in an episode, update times $\mathcal{L} $ and batch size $B$\;
    \For{Agent $k \in \mathcal{N} \cap \mathcal{M}$}
    {
        Initialize $\pi_{\theta_k}$, $V_{\omega_{k}}$, $\alpha$, $\beta$, $\gamma$ and $\epsilon$.
    }
    \For{Episode $e \in 1,\ldots,E $}
    {
    Reset Stackelberg game environment state $S_0$ and clear up the replay buffer $\mathcal{B F}_k$ \;
    \For{Time slot $t \in 1,\ldots,T$}
    {
    Input $o_j^{t}$ into MRP's actor policy $\pi_{\boldsymbol{\theta}_j}$ and determine the current price strategy $p_j^t$\;
    Input $o_i^{t}$ into MSP's actor policy $\pi_{\boldsymbol{\theta}_i}$ and determine the current bandwidth demand strategy $b_{i}^{t}$\;
    Update $S_t$ to $S_{t+1}$ and calculate rewards $R_j^t$ for MRP $j$ and $R_i^t$ for MSP $i$ through (\ref{MRP}) and (\ref{MSP}).\;
    Store transition $(o_k^t,a_k^t,R_k^t,o_{k+1}^t)$ into
    $\mathcal{B F}_k$\;
    \If{ $t\%\left|B\right|==0$}
    {
    \For{$l \in 1,\ldots,\mathcal{L}  $}
    {
    \For{Agent $k \in \mathcal{N} \cap \mathcal{M}$}{
    Sample a random mini-batch of data with a size $\left|B\right|$ from $\mathcal{B F}_k$\;
    Calculate the loss using (\ref{actorloss}) and (\ref{criticloss})\;
    Update $\nabla_{\theta_k}L_\pi$ and $\nabla_{{\omega}_k}L_{V}$ via (\ref{actorupdate}) and (\ref{criticupdate})\;
    }
    }
    }
    }
    }
\end{algorithm}

\subsection{Algorithm Design}
As illustrated in Fig. \ref{MALPPO}, the Stackelberg game environment consists of $M$ MRPs serving as leaders and $N$ MSPs acting as followers. Each player is an intelligent agent employing a DRL controller to make decisions. The MRP agent utilizes its DRL controller to determine the bandwidth selling price, while the MSP agent leverages its DRL controller to decide the bandwidth amount to purchase from each MRP. The detailed components of the DRL controller are illustrated on the right side of Fig. \ref{MALPPO}. Each controller $k$ comprises an actor network $\pi_{\theta_k}$, a critic network $V_{\omega_{k}}$, a replay buffer $\mathcal{BF}_k$, policy optimizer, and value optimizer. Here, $\theta_k$ and $\omega_{k}$ denote the parameters of the actor and critic, respectively.

Proximal Policy Optimization (PPO) stands as a model-free, on-policy, and policy-based algorithm proposed by OpenAI in 2017, and it has demonstrated excellent performance in many resource allocation tasks \cite{xu2021multiagent,zhan2020learning}. In \cite{yu2022surprising}, the Multi-Agent Proximal Policy Optimization (MAPPO) is introduced to extend PPO to multi-agent reinforcement learning problems. However, MAPPO assumes information sharing among agents and employs a training approach called Centralized Training and Decentralized Execution (CTDE). During training, it maintains a global critic network to evaluate the state values of all agents. This training method is not suitable for our privacy protection scenario. Therefore, our approach adopts a fully decentralized training method with a decentralized actor and a decentralized crtitc for each agent.

Moreover, to expedite convergence and alleviate the risk of getting stuck in local optima during training, we employ Long Short-Term Memory (LSTM) networks as the foundational architecture for both the actor and critic. Both the actor and the critic network consist of two multi-layer perceptrons (MLPs). Local observation variables serve as input to the MLP, extracting feature vectors that are subsequently fed into the LSTM network with two hidden layers. The memory units within the LSTM network dynamically manage information retention and forgetting. Finally, the output is obtained through two MLPs for actions or state values. Previous works have demonstrated the effectiveness of recurrent neural networks like LSTM in handling POMDP environments \cite{zhan2020learning}.

Next, we provide a detailed explanation of the Policy Optimization method and the update method for the value function. The policy iteration objective is defined as
\begin{equation}
    \label{actorloss}
    L_{\pi}\left(\boldsymbol{\theta}_{k}\right)=\hat{\mathbb{E}}_{t}\left[\min \left(r_{t}\left(\boldsymbol{\theta}_{k}\right) A\left(o_{t}, a_{t}\right), f_{c l i p}\left(r_{t}\left({\theta}_{{k}}\right)\right) A\left(o_{t}, a_{t}\right)\right)\right]
\end{equation}
where
\begin{equation}
    r_{t}({\theta}_{k})=  \frac {\pi _ {{\theta}_{k} }(a_ {t}|o_ {t})}{\pi _ {{\theta}_{k}^{old} }(a_ {t}|o_ {t})}  ,
\end{equation}
\begin{equation}
    \begin{split}
        A\left(o_{t}, a_{t}\right)=&-V_{\pi_{\boldsymbol{\omega }_{k}}}\left(o_{t}\right)+\sum_{l=t}^{T-1} \gamma^{l-t} R(o_t,a_t)
        +\gamma^{T-t} V_{\pi_{\boldsymbol{\omega}_{k}}}\left(o_{t}\right)
    \end{split}
\end{equation}
\begin{equation}
    f_{clip}(r_{t}({\theta}_{k}))=  \begin{cases}1-\epsilon,\:r_{t}({\theta}_{k})<1-\epsilon, \\
        1+\epsilon,\:r_{t}({\theta}_{k})>1+\epsilon, \\r_{t}({\theta}_{k}),\:1-\epsilon \leq r_{t}({\theta}_{k}) \leq 1+\epsilon.\end{cases}
\end{equation}

Here the $\hat{\mathbb{E}}_t[\ldots]$ indicates the empirical average reward over a batch of samples, $r_{t}({\theta}_{k})$ is the importance ratio,
$A\left(o_{t}, a_{t}\right)$ is an estimator of the advantage function using \textit{Generalized Advantage Estimation} method and $\epsilon$ is a hyperparameter. The policy network $\pi_{\theta_k}$ is updated through mini-batch stochastic gradient ascent method as follows:
\begin{equation}
    \label{actorupdate}
    \theta_k\leftarrow\theta_k+l_{\alpha}\nabla_{\theta_k}L_\pi
\end{equation}
where $l_{\alpha}$ is the learning rate for policy network. The loss function for updating the critic network is defined as
\begin{equation}
    \label{criticloss}
    L_{V}\left(\boldsymbol{\omega}_{k}\right)=\hat{\mathbb{E}}_{t}[V_{\omega_{k}}\left(o_{t}\right)-V_{t}^{\text {targ }}]^{2}
\end{equation}
where $V_t^{\text {targ }}$ denotes the total discounted reward from time step $t$ until the end of the episode. The critic network is updated through the mini-batch stochastic gradient descent method as follows:
\begin{equation}
    \label{criticupdate}
    \omega_k\leftarrow\omega_k-l_{\beta} \nabla_{{\omega}_k}L_{V}
\end{equation}
where $l_{\beta}$ represents the learning rate for the critic network. Building on the aforementioned analysis, we present the proposed MALPPO algorithm in \textbf{Algorithm 2}. The algorithm's complexity is influenced by the multiplication operations in the MLP and the memory and forget operations in the LSTM cell. In our algorithm, the input dimension is determined by the number of MSPs $N$ and the number of MRPs $M$. Let $H$ represent the number of neurons in the hidden layers of the MLPs and LSTM. The complexity analysis of MLP is $O((N+M)H)$ and the complexity analysis of LSTM is $O(H^{2})$. Consequently, the overall complexity analysis is $O((N+M)H + H^{2})$ for each time slot.


\section{Simulation Results}

\subsection{Simulation Settings}

We consider that there are $M\in \left [ 2,4 \right ]$ MRPs and $N\in \left [ 2,6 \right ] $ MSPs in the VT migration. Each VT is characterized by a data size $D_i\in \left [ 10,50 \right ](\rm{MB}) $, required CPU cycles $L_i \backsim \mathcal{N}(5000,500)  (\rm{Megacycles})$, and a maximum tolerant delay $K_i^{max} \in \left [ 2,4 \right ] (\rm{s})$. The satisfaction coefficient and sensitivity coefficient of MSP $\alpha_i$ and $\beta_i$ follow a normal distribution $\mathcal{N}(\mu_{\alpha},1)$ and $\mathcal{N}(\mu_{\beta},1)$, respectively. The social ties between MSPs follow a normal distribution $\mathcal{N}(\mu_w,1)$ and the computation capability of MSP $i$ $f_i \backsim \mathcal{N}(15,5) (\rm{GHZ})$, we set $\mu_{\alpha}$ = 30, $\mu_{\beta}$ = 30, $\mu_w$ = 5 \cite{nie2020multi}. For MRP $j$, the task arriving rate $\lambda_{j}$ and the task processing rate $\mu_{j}$ follow a normal distribution $\mathcal{N}(450,20) (\rm{PPS})$ and $\mathcal{N}(500,20) (\rm{PPS})$, respectively. The cost of executing VT migration $c_j$ follows a normal distribution $\mathcal{N}(\mu_c,0.05)$ and the maximum price $p^{max}$ is 1.5. We set $\mu_c = 0.1$ by default. Regarding the RSU transmission parameters, the transmitter power $\rho $ is $40\rm{dBm}$, the unit channel power gain $h$ is $-20\rm{dB}$, the distance between the RSUs $d$ is $500\rm{m}$, the path-loss coefficient $\epsilon$ is $2$, and the average noise power $N_0$ is $-150\rm{dBm}$ \cite{zhang2023learningbased}.

We conducted our experiments using PyTorch 2.1 on Windows 11. The parameters of the DRL were fine-tuned with specific values: the number of experiences $L=3$, minibatch size $B=20$, number of training episodes $E=20$, time slot of each episode $T=100$, update times $\mathcal{L}=5$, and $l_{\alpha}=l_{\beta}=0.0001$ during experiments. Both hidden layers of the MLP have $128$ nodes, and both hidden layers of LSTM have $64$ nodes. The Adam optimizer is employed for optimizing the neural networks.
\begin{figure}
    \centering
    \begin{subfigure}{\columnwidth}
        \centering
        \includegraphics[width=0.95\linewidth]{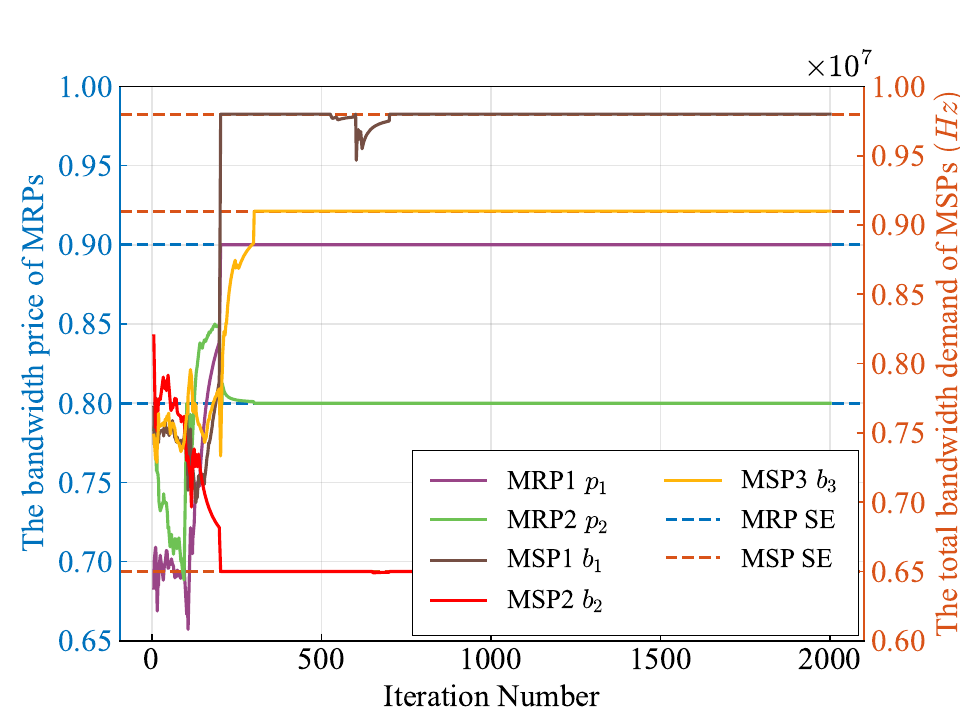}
        \caption{MSPs' bandwidth demand and MRPs' pricing strategies}
        \label{Convergence_act}
    \end{subfigure}

    \vspace{1em} 

    \begin{subfigure}{\columnwidth}
        \centering
        \includegraphics[width=0.95\linewidth]{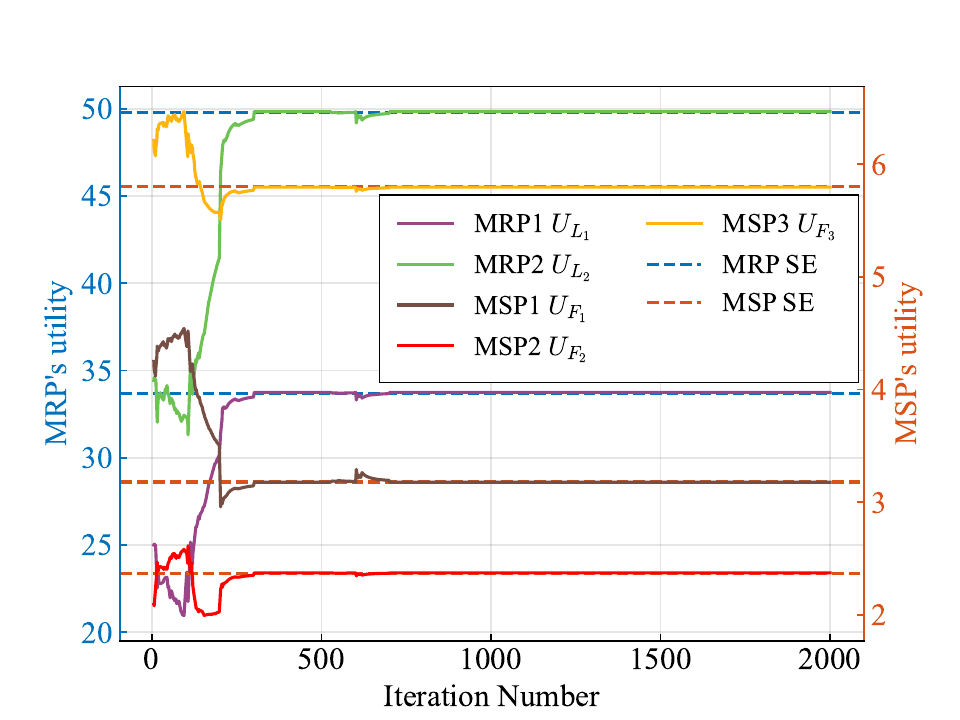}
        \caption{MSPs' and MRPs' utilities}
        \label{Convergence_utility}
    \end{subfigure}

    \caption{The convergence analysis of the proposed MALPPO, with $N$ = 3 and $M$ = 2, the migration costs are $c_1 = 0.3$ and $c_2 = 0.1$. The satisfaction coefficients are $\alpha_1 = 30$, $\alpha_2 = 25$, and $\alpha_3 = 35$.}
    \label{Convergence}
\end{figure}
In the convergence analysis, we compare our MALPPO algorithm with six baseline methods as follows:
\begin{itemize}
    \item MALPPO-RL: the degraded version of the proposed MALPPO scheme, where the leader MRPs randomly select their pricing strategies.
    \item MALPPO-RF: the degraded version of the proposed MALPPO scheme, where the follower MSPs randomly select their bandwidth demand strategies.
    \item MAPPO \cite{yu2022surprising}: the PPO algorithm for multi-agent environment.
    \item MAA2C \cite{li2018distributional}: the A2C algorithm for multi-agent environment.
    \item ADMM-SE: we leverage ADMM in \textbf{Algorithm 1} to obtain Stackelberg Equilibrium under information sharing.
    \item Random: all players randomly choose their pricing strategies or demand strategies in the VT migration.
\end{itemize}
\begin{figure}
    \centering
    \includegraphics[width=0.95\linewidth]{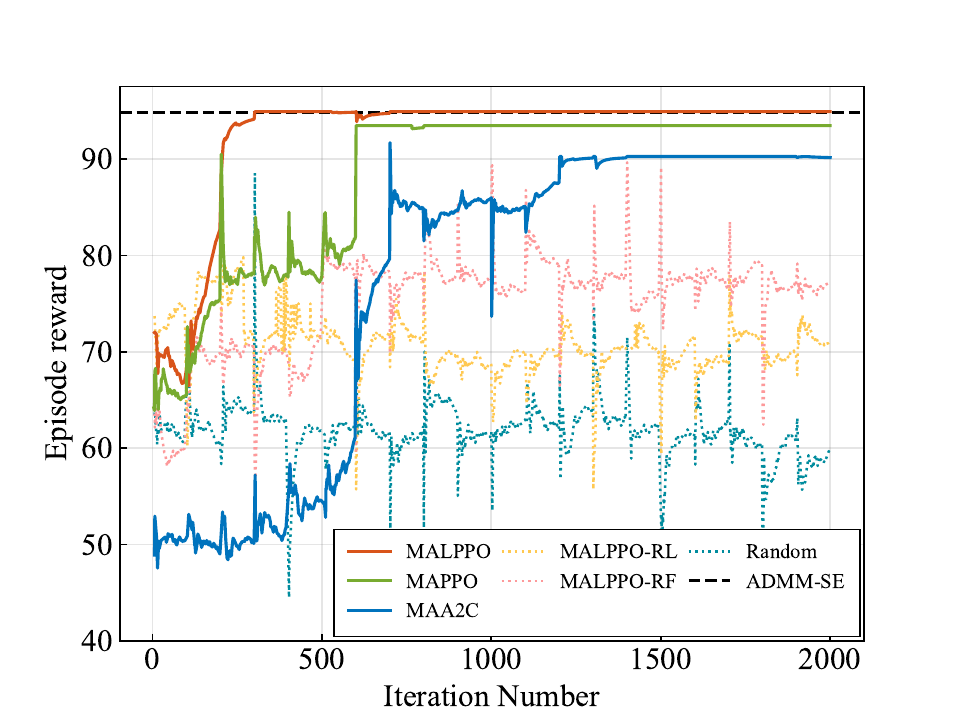}
    \caption{Comparison of episode reward curves of
        MALPPO and baselines for the Stackelberg Game.}
    \label{compare}
\end{figure}
\begin{figure*}
    \centering

    \begin{subfigure}{0.24\textwidth}
        \centering
        \includegraphics[width=\linewidth]{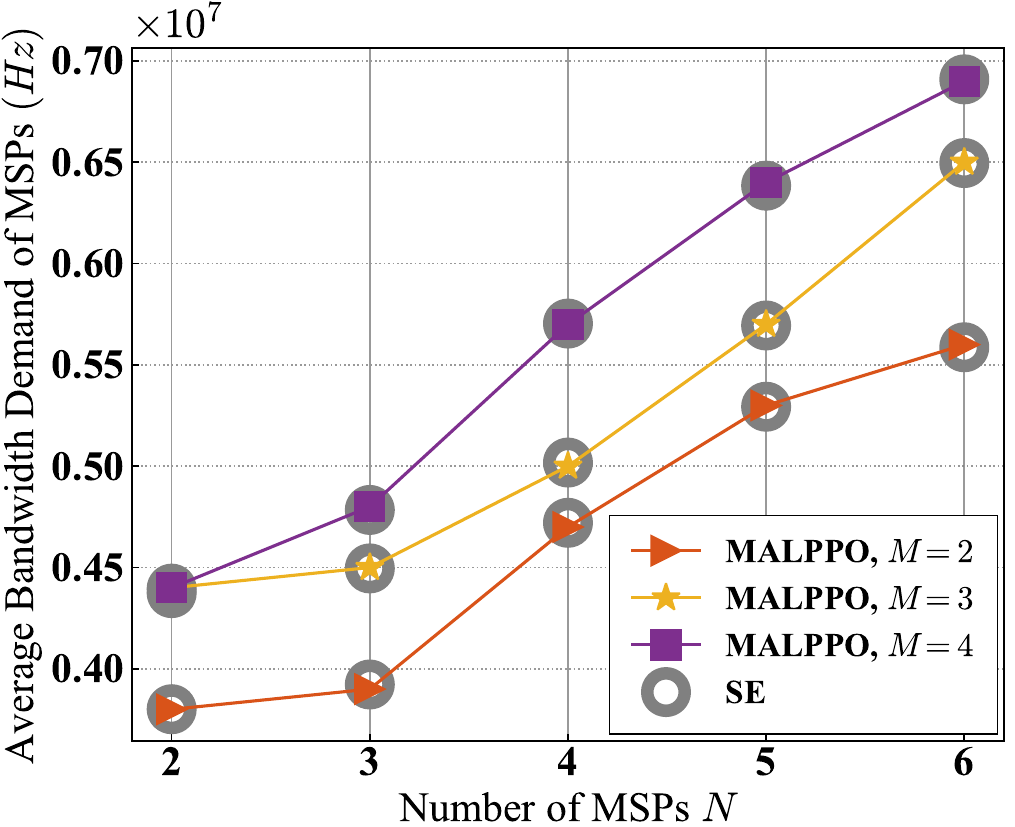}
        \caption{Average Demand of the MSPs}
        \label{Fig5_a}
    \end{subfigure}
    \hfill
    \begin{subfigure}{0.24\textwidth}
        \centering
        \includegraphics[width=\linewidth]{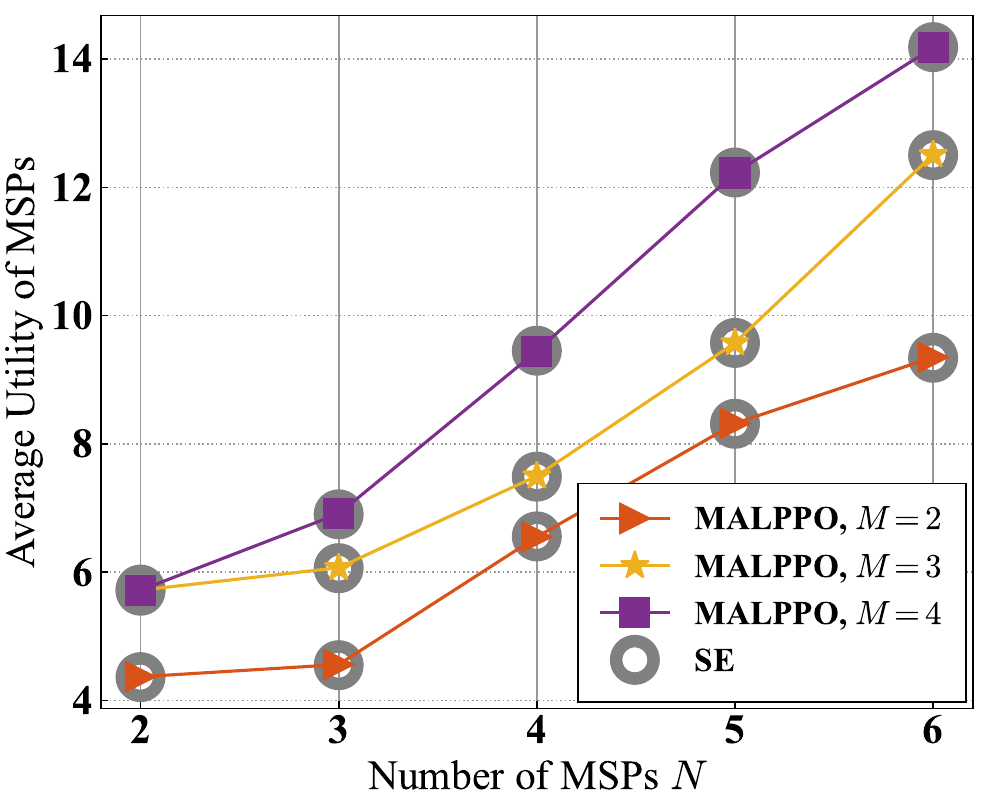}
        \caption{Average Utility of the MSPs}
        \label{Fig5_b}
    \end{subfigure}
    \hfill
    \begin{subfigure}{0.24\textwidth}
        \centering
        \includegraphics[width=\linewidth]{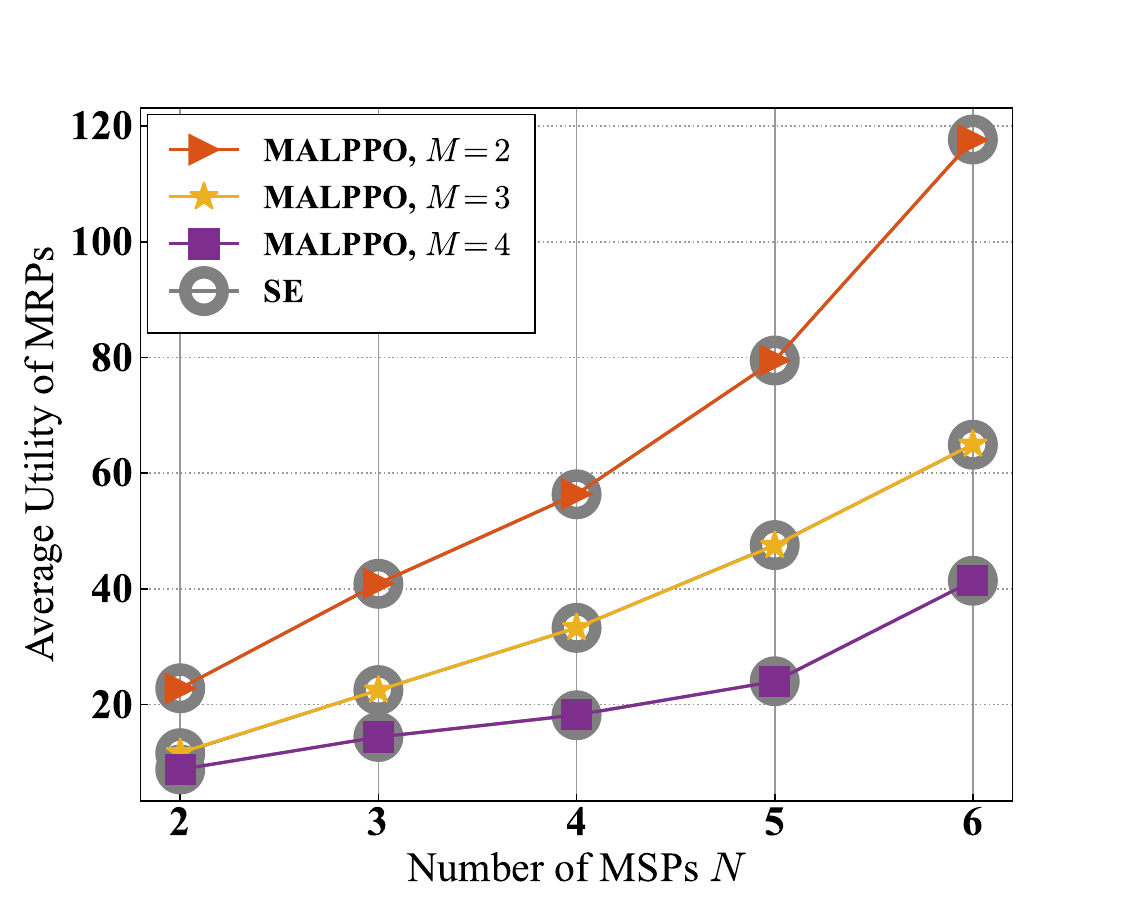}
        \caption{Average Utility of the MRPs}
        \label{Fig5_c}
    \end{subfigure}
    \hfill
    \begin{subfigure}{0.24\textwidth}
        \centering
        \includegraphics[width=\linewidth]{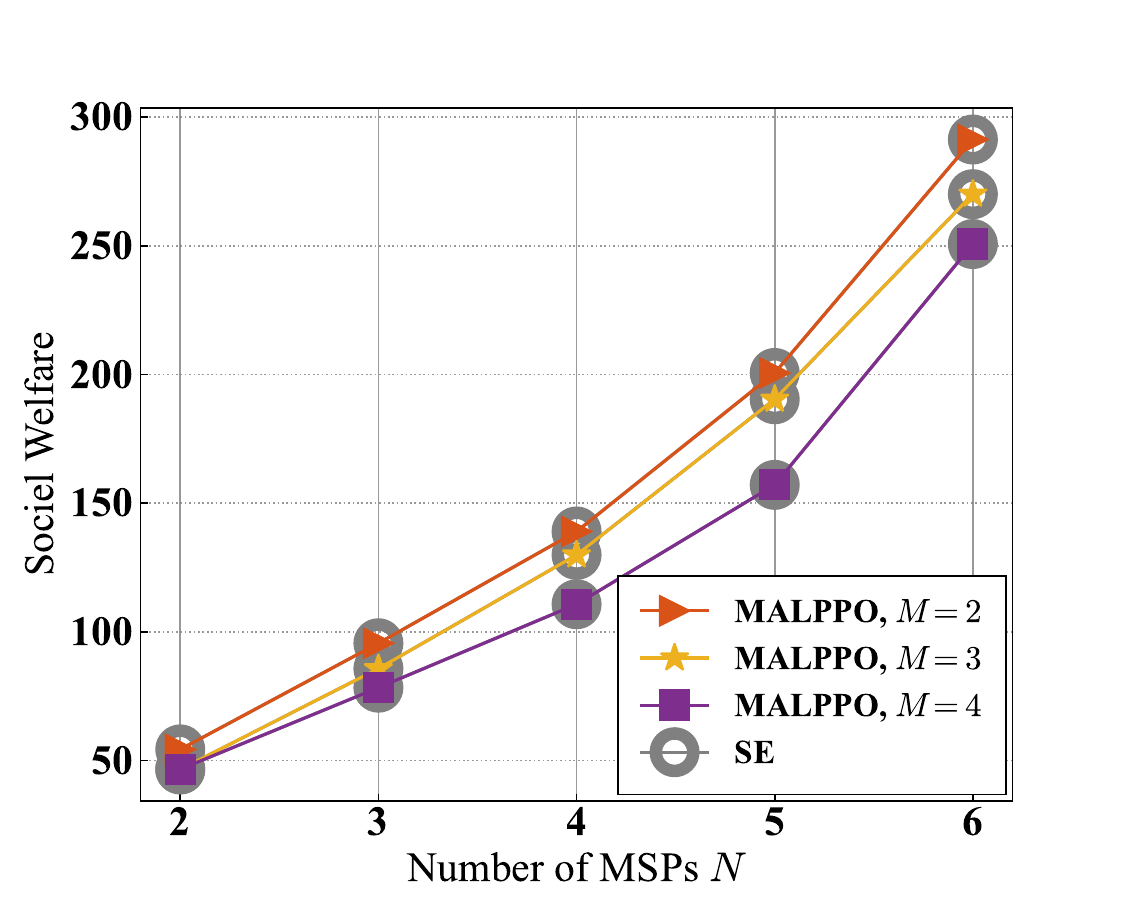}
        \caption{Social Welfare}
        \label{Fig5_d}
    \end{subfigure}

    \caption{
        The performance of the MALPPO is evaluated as the number of MSPs ($N$) varies within$\left[2, 6\right]$, and the number of MRPs is considered in the set $\left[2, 3, 4\right]$. The social welfare is calculated as the sum of the utilities of MRPs and MSPs.}
    \label{Fig5}
\end{figure*}

\begin{figure*}
    \begin{subfigure}{0.24\textwidth}
        \includegraphics[width=\linewidth]{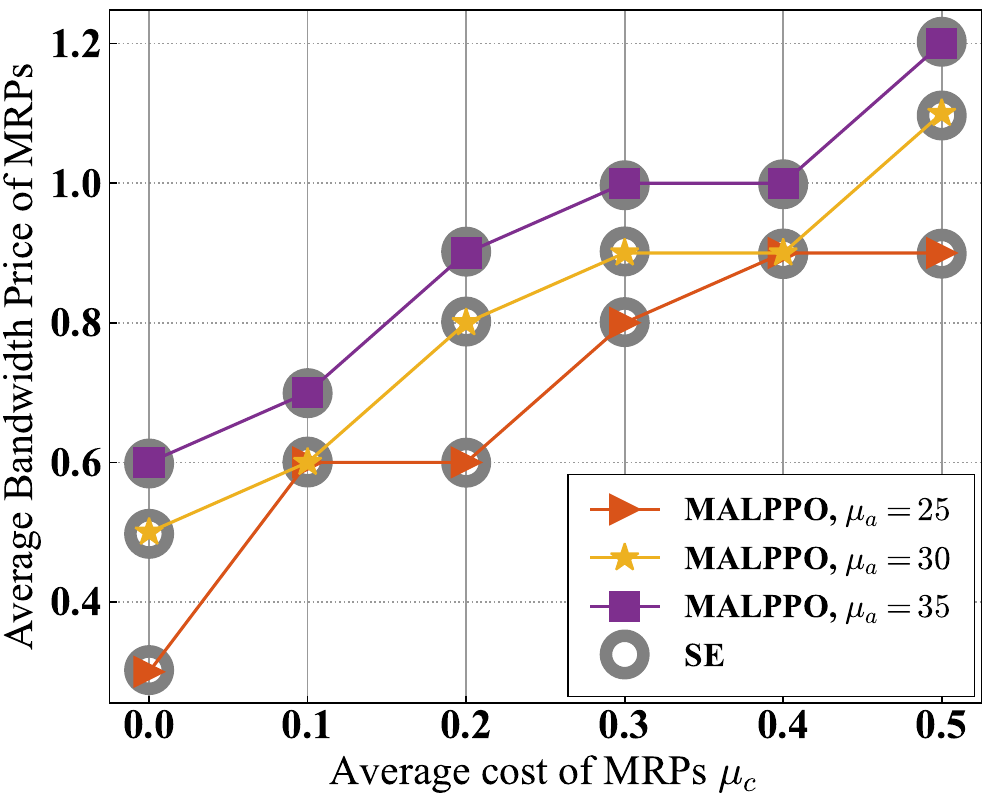}
        \caption{Average Price of the MRP}
        \label{Fig6_a}
    \end{subfigure}
    \hfill
    \begin{subfigure}{0.24\textwidth}
        \includegraphics[width=\linewidth]{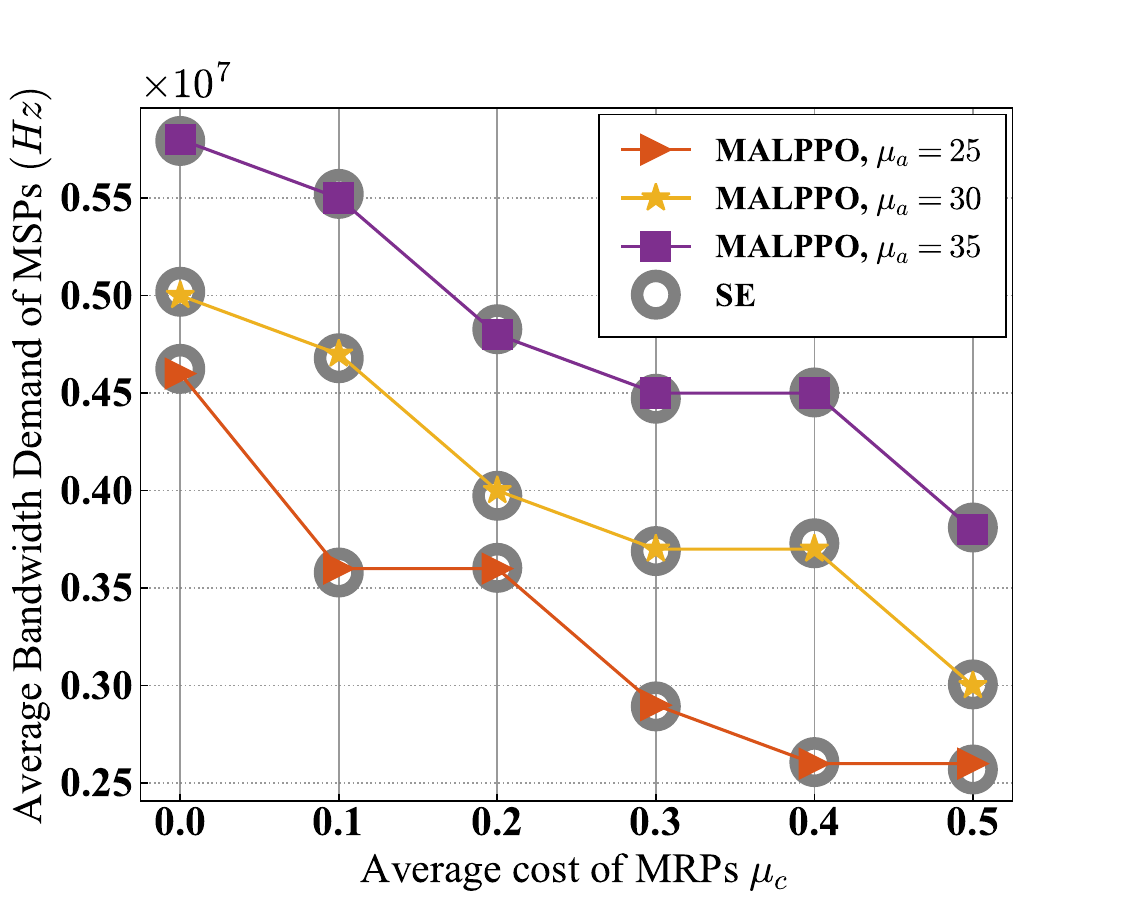}
        \caption{Average Demand of the MSPs}
        \label{Fig6_b}
    \end{subfigure}
    \hfill
    \begin{subfigure}{0.24\textwidth}
        \includegraphics[width=\linewidth]{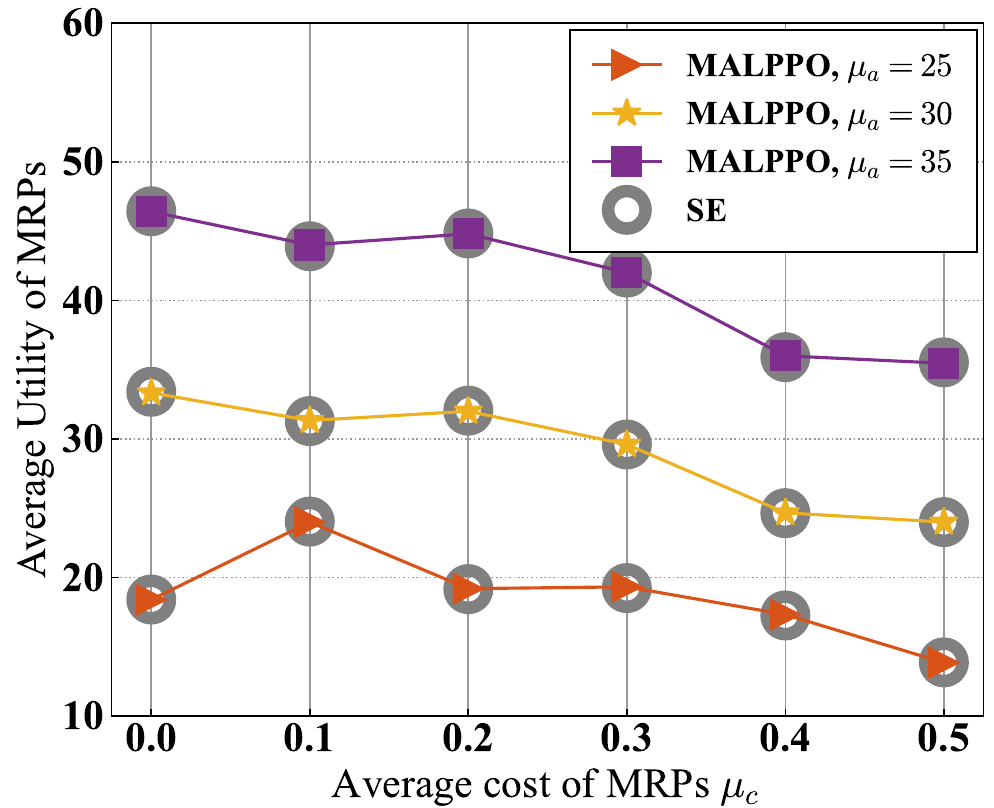}
        \caption{Average Utility of the MRPs}
        \label{Fig6_c}
    \end{subfigure}
    \hfill
    \begin{subfigure}{0.24\textwidth}
        \includegraphics[width=\linewidth]{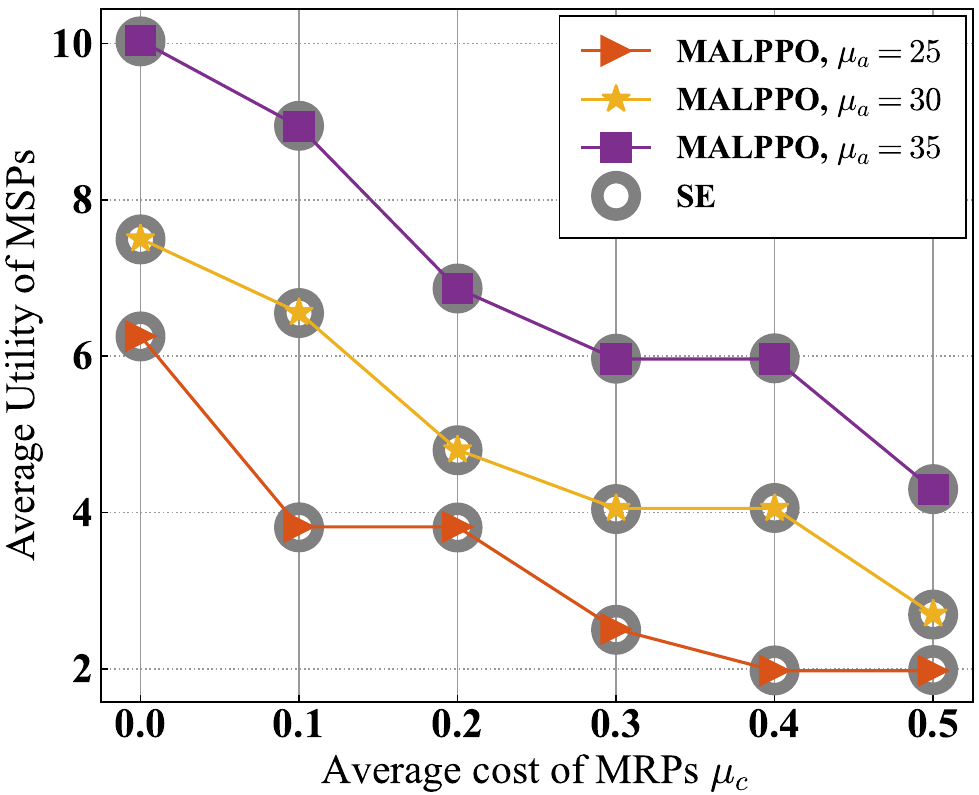}
        \caption{Average Utility of the MSPs}
        \label{Fig6_d}
    \end{subfigure}

    \caption{
        The performance of the MALPPO is assessed as the average cost coefficient of MRPs ($\mu_c$) varies within $\left[0.0, 0.5\right]$, the average satisfaction coefficient of MSPs ($\mu_a$) is selected from the set $\left[25, 30, 35\right]$, with $M = 3$  and $N = 4$.}
    \label{Fig6}
\end{figure*}

\subsection{Convergence Analysis}
We initiate our analysis by examining the convergence behavior of the proposed MALPPO algorithm with 2 MRPs and 3 MSPs. The migration costs of MRPs are $c_1 = 0.3$ and $c_2 = 0.1$, while the satisfaction coefficients of MSPs are $\alpha_1 = 30$, $\alpha_2 = 25$ and $\alpha_3 = 35$. As illustrated in Fig. \ref{Convergence}, it is evident that during the training iterations of the MALPPO algorithm, both the actions and utilities of MSPs and MRPs converge to the SE after approximately 200 iterations. Specifically, the bandwidth prices of MRPs converge to 0.9 and 0.8, while the total bandwidth demands of MSPs converge to 0.60, 0.79, and 0.91. The utilities of MRPs and MSPs under SE are determined to be 30.56, 45.12, 3.78, 1.66, and 5.52, respectively.

In the subsequent analysis, we conduct a comparative performance evaluation of the MALPPO algorithm against six baseline approaches under identical experimental settings. Fig. \ref{compare} shows that our proposed MALPPO algorithm significantly outperforms other algorithms with the highest reward and the fastest convergence speed. Notably, the random approach consistently yields the lowest reward, and both MALPPO-RL and MALPPO-RF, which incorporate random pricing and demand strategies, achieve inferior rewards compared to the fully DRL algorithm. In terms of reward improvement, our MALPPO algorithm outperforms the vanilla MAPPO method by 1.5\%, the MAA2C method by 5.3\%, the MALPPO-RF method by 23.2\%, the MALPPO-RL method by 35.6\%, and the Random method by a substantial 58.2\%. Concerning convergence speed, our algorithm converges in approximately 200 iterations, whereas MAPPO and MAA2C require around 600 and 1200 iterations, respectively.

As shown in Figs. \ref{Convergence}-\ref{Fig7}, we compare our algorithm with the SE solved by ADMM in both convergence analysis experiments and parameter influence analysis experiments. In Fig. \ref{Convergence} and \ref{compare}, SE is represented by dashed lines, while in Figs. \ref{Fig5}-\ref{Fig7}, SE is indicated by a gray circle. It can be observed that under various parameter conditions, our algorithm consistently converges accurately to the SE.

\subsection{Impact of MSPs' number $N$}
In Fig. \ref{Fig5}, the influence of the MSPs' number $N$ on system performance is demonstrated, with $N$ varying within $\left[2,6\right]$. As shown in Fig. \ref{Fig5_a} and Fig. \ref{Fig5_b}, an increase in the number of MSPs results in a corresponding rise in both the average bandwidth demand and average utilities of MSPs. For instance, in a scenario with only two MRPs, as the number of MSPs increases from 2 to 6, the average bandwidth purchase of MSPs rises from 0.38 * 10e7 Hz to 0.56 * 10e7 Hz, and the average utility of MSPs increases from 4.37 to 9.35. This observed trend can be attributed to several factors. On the one hand, due to the social effect, as the number of MSPs increases, MSPs are more inclined to purchase more bandwidth to expedite migration, aiming for higher utility and service experience. Hence, the demand for purchased bandwidth rises. On the other hand, despite the increase in bandwidth demand from MSPs, the competition among MRPs suppresses the increase in bandwidth prices. If one MRP raises its bandwidth price, more MSPs may choose to purchase additional bandwidth from other MRPs. Consequently, the utility of that MRP decreases. Thus, as the number of MSPs increases, the prices among MRPs remain relatively stable, leading to an enhancement in MSPs' utilities. As the total demand for bandwidth increases, as illustrated in Fig. \ref{Fig5_c} and Fig. \ref{Fig5_d}, the average utility of MRPs and social welfare also rise.
\subsection{Impact of MRPs' number $M$}
In Fig. \ref{Fig5}, the influence of the MRPs' number $M$ on system performance is analyzed, with $M$ varying from 2 to 4. The figure reveals that with an increase in the number of MRPs, there is a notable rise in the average bandwidth purchase of MSPs and an associated increase in MSPs' utility. Conversely, the average utility of MRPs decreases, leading to a slight decline in social welfare. For instance, with 3 MSPs, when the number of MRPs is 2, the average utility for MRPs is 40.95, and the social welfare is 95.58. However, when the number of MRPs increases to 4, the average utility for MRPs drops to 14.4, and the social welfare decreases to 78.33. This observed phenomenon can be attributed to heightened competition among MRPs as the number of MRPs increases. To attract more MSPs to purchase their bandwidth resources, each MRP chooses to lower its bandwidth prices. This competition ultimately leads to a scenario where all MRPs reduce their prices. MSPs, benefiting from lower prices, can afford to purchase more bandwidth to expedite migration, thus increasing their own utility. However, for MRPs, the intensified competition results in decreased utility, contributing to a reduction in overall social welfare.
\begin{figure}

    \begin{subfigure}{\columnwidth}

        \includegraphics[width=0.90\linewidth]{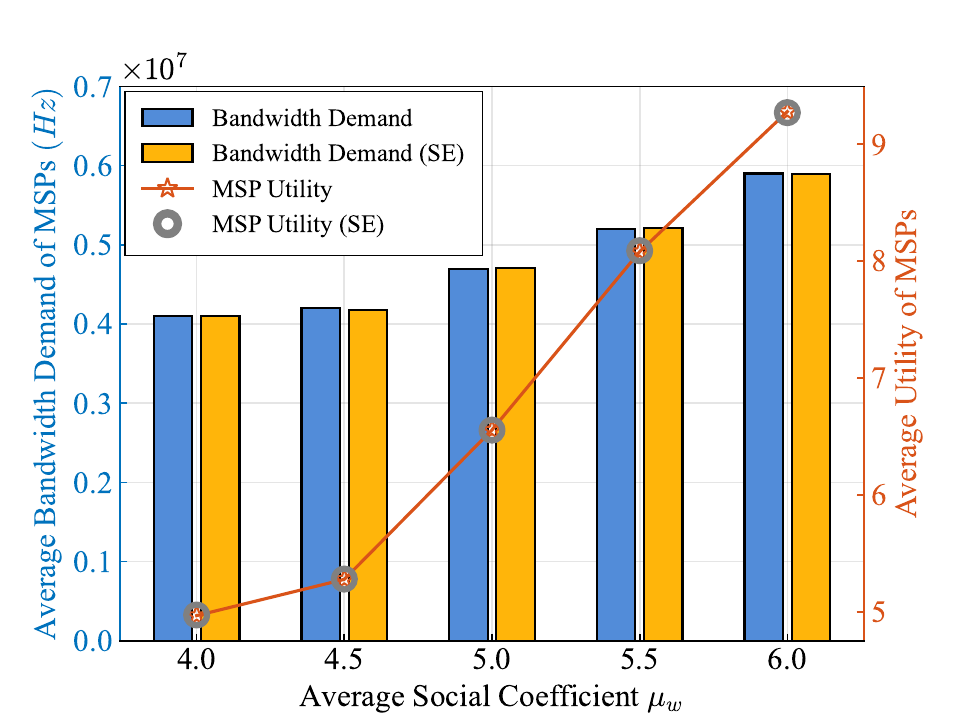}
        \caption{Average Bandwidth Demand and Utility of the MSPs}
        \label{fig:act}
    \end{subfigure}

    \vspace{1em} 

    \begin{subfigure}{\columnwidth}

        \includegraphics[width=0.90\linewidth]{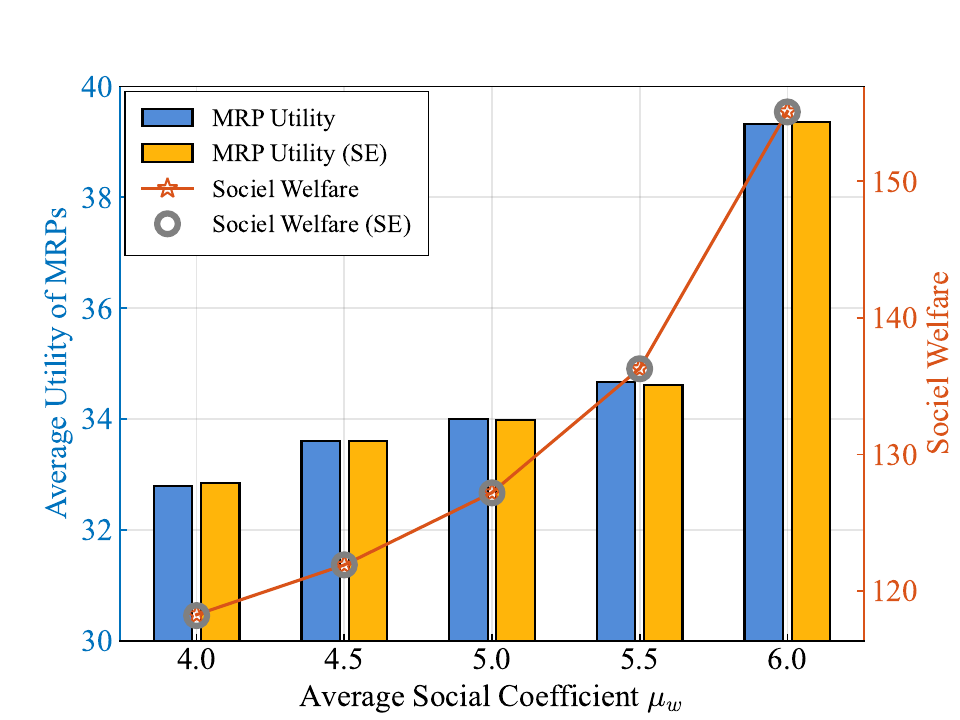}
        \caption{Average Utility of the MRPs and Social Welfare}
        \label{fig:reward}
    \end{subfigure}

    \caption{The performance of the MALPPO is evaluated by varying the average social coefficient within $\left[4.0, 6.0\right]$.}
    \label{Fig7}
\end{figure}

\subsection{Impact of the average cost coefficient $\mu_c$}
In the analysis experiment concerning the impact of MRP cost coefficients and MSP satisfaction coefficients, we set the number of MRPs as 3, the number of MSPs as 4, and the average cost coefficient of MRPs $\mu_c$ varies within the range $\left[0.0, 0.5\right]$. In Fig. \ref{Fig6}, we observe that as the average migration cost for MRPs increases, the bandwidth prices set by MRPs also exhibit an upward trend to avoid the loss of profit. For instance, in Fig. \ref{Fig6_a}, when $\mu_{\alpha} = 30$, the cost increases from 0.1 to 0.5, and the bandwidth prices set by MRPs rise from 0.6 to 1.1. Therefore, as the prices set by MRPs increase, MSPs tend to reduce the demand for purchased bandwidth to maintain their own utility. Consequently, the average utilities of both MSPs and MRPs decrease, which decreases social welfare.

\subsection{Impact of the average satisfaction coefficient $\mu_a$}
The satisfaction coefficient $\alpha$ reflects an MSP's satisfaction level regarding obtaining a certain amount of bandwidth for VT migration. In Fig. \ref{Fig6}, we set three groups of average satisfaction coefficients $\mu_w$, namely 25, 30, and 35. It can be observed that as the MSP satisfaction coefficient increases, both the average price set by MRPs and the average bandwidth purchased by MSPs increase. Additionally, the average utilities of both MSPs and MRPs also increase. For instance, when the average cost is 0.2, at a satisfaction level of 25, the average price set by MRPs is 0.6, and the average bandwidth purchased by MSPs is 0.36 * 10e7 Hz. However, at a satisfaction level of 35, the price increases to 0.9, and the bandwidth purchase quantity rises to 0.48 * 10e7 Hz. The reason is that as $\mu_a$ increases, the higher satisfaction values incentivize MSPs to purchase more bandwidth resources to expedite VT migration. Simultaneously, the entire group of MRPs, in response to the increased tolerance of MSPs for prices and their substantial bandwidth purchases, may choose to raise prices to further enhance their own revenue. Consequently, the utility of MRPs increases, leading to an increase in the utility of MSPs as well.
\subsection{Impact of the average social coefficient $\mu_w$}

The average social coefficient $\mu_w$ reflects the cooperative nature among the MSPs in the group. In Fig. \ref{Fig7}, we analyze this social benefit by varying $\mu_w$ in increments of 0.5, ranging from 4.0 to 6.0, with the number of MSPs $N=4$ and the number of MRPs $M=3$. It is observed that as $\mu_w$ decreases, both the average bandwidth purchase quantity and average utility of MSPs exhibit a decline. Additionally, the utility of MRPs and social welfare also decreases. The reason is that with the decrease in the social coefficient, each MSP's ability to gain additional revenue from other MSPs diminishes. Consequently, external social benefits decrease, leading MSPs to reduce their bandwidth purchases. Due to the decrease in bandwidth purchase quantity, the utility of MRPs also declines, resulting in a decrease in social welfare.
\begin{figure}
    \includegraphics[width=0.90\linewidth]{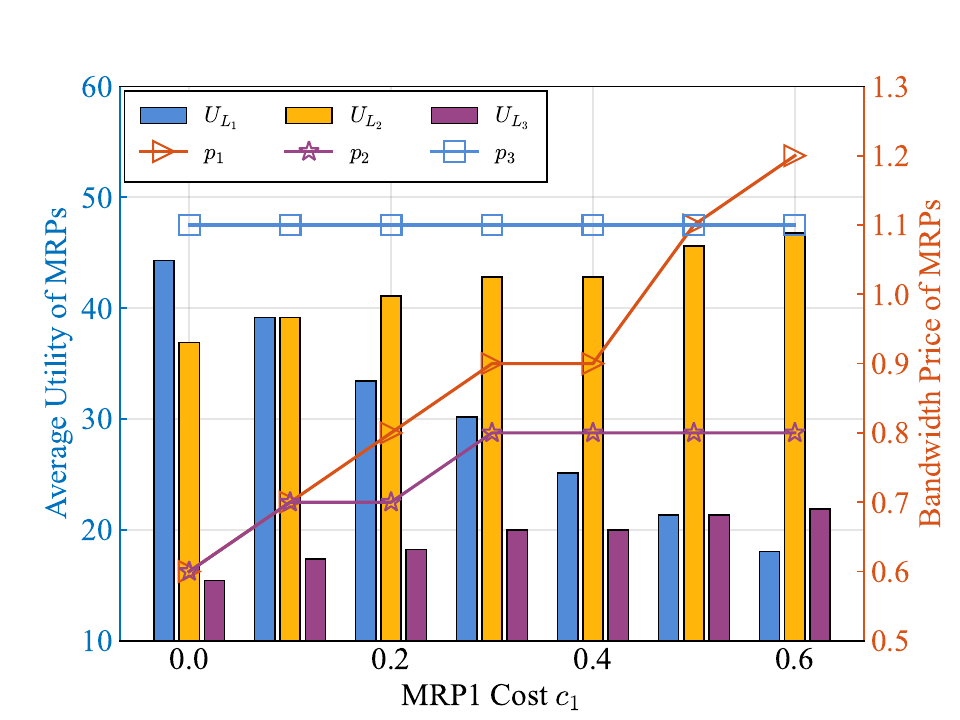}
    \caption{
        The analysis assesses the impact of MRP1's cost ($c_1$) within $\left[0.0, 0.6\right]$, with $c_2 =0.1$, $c_3 = 0.5$, and $N = 4$.}
    \label{Fig8}
\end{figure}
\subsection{Impact of the MRP1's cost $c_1$}
We delve into the analysis of the competitive dynamics among MRPs by varying the cost of one MRP. Specifically, we set the migration costs as $c_1 = 0.0$, $c_2 = 0.1$, $c_3 = 0.5$, and the number of MSPs $N = 4$. We perform the analysis by changing $c_1$ from 0.0 to 0.6. In Fig. \ref{Fig8}, it can be observed that as the MRP cost $c_1$ increases, MRP1 is compelled to raise the price to maintain profitability. Consequently, the bandwidth sales of MRP1 decrease and the utility of MRP1 continuously declines. For MRP3, since its own cost remains constant and the cost of MRP1 keeps rising, the probability of MSPs choosing MRP3 gradually increases, leading to an increase in the utility of MRP3. As for MRP2, initially engaged in intense competition with MRP1, it adopts a slightly higher pricing strategy to ensure a modest increase in profit. Subsequently, with unchanged pricing and an increasing probability of being chosen by MSPs, the utility of MRP2 continues to grow.

\section{Conclusion}

In this paper, we propose a novel incentive mechanism between MSPs and MRPs for the Migration-empowered Vehicular Metaverses.
Firstly, we integrated the social effects among MSPs and the competitiveness among MRPs to design a Stackelberg game with multi-leader multi-follower. Subsequently, we employed the backward induction method to prove the existence and uniqueness of the Stackelberg Equilibrium and utilized the ADMM algorithm to obtain specific equilibrium solutions. Following that, we introduced privacy protection requirements and developed the MALPPO algorithm based on LSTM and PPO to find optimal solutions in the presence of incomplete information.
\bibliographystyle{IEEEtran}
\bibliography{ref}

\end{document}